\documentclass[11pt]{article}
\usepackage{amsmath, amssymb, amsthm}
\usepackage{cite}
\usepackage{graphicx}
\usepackage{subcaption}
\usepackage[hidelinks]{hyperref}
\usepackage{xcolor}
\usepackage{array}
\usepackage{textcomp}

\input{margins.sty}
\input{variables.sty}
\input{commands.sty}

\title{Fundamental limitations to no-jerk gearshifts of multi-speed transmission architectures in electric vehicles}
\author{Marc-Antoine Beaudoin \thanks{Corresponding author: ma.beaudoin@mail.mcgill.ca} \thanks{Intelligent Automation Lab, Centre for Intelligent Machines, McGill University, Room 503, McConnell Engineering Building, 3480 University Street, Montreal, QC, Canada, H3A 0E9} \and Benoit Boulet \footnotemark[2]}
\date{\today}

\begin{document}
\maketitle

\begin{abstract}
Multi-speed transmissions can enhance the performance and reduce the overall cost of an electric vehicle, but they also introduce a challenge: avoiding gearshift jerk, which may sometimes prove to be impossible in the presence of motor and clutch saturation. In this article, we introduce three theorems that explicitly define the fundamental limitations to no-jerk gearshifts resulting from motor or actuator saturation. We compare gearshifts that consist of transferring transmission torque from one friction clutch to another, to the case in which one of the clutches is a one-way clutch. We show that systems with a one-way clutch are more prone to motor saturation, causing gearshift jerk to be more often inevitable. We also study the influence of planetary gearsets on the gearshift dynamical trajectories, and expose the impact on the no-jerk limitations. This work offers tools to compare transmission architectures during the conceptual design phase of a new electric vehicle. \\

\end{abstract}

\section{Introduction}
Uninterrupted gearshifts are a desirable feature of electric vehicles, but not all multi-speed transmissions are capable of it. To provide this capability, vehicle design engineers must select a transmission architecture in which the motor torque can be continuously transferred from one transmission path to another during gearshifts. But in the presence of motor and clutch saturation, even such a transmission can fail to provide an uninterrupted gearshift. In this article, we explore the fundamental limitations on gearshift performance that originate from actuator saturation.
\\

These fundamental limitations should be considered early in an electric vehicle design process, such as when selecting the transmission type during the conceptual design phase. Established design methodologies attribute a high importance to the conceptual design phase, as its outcome has a large influence on the rest of the design project, and ultimately, the product quality~\cite{ulrich_product_2012,pahl_engineering_2007,french_conceptual_1999}. We wish to provide electric vehicle design engineers with clear expectations on the potential gearshift performance of various transmission architectures, before they delve into resource-intensive detailed modelling.

\subsection{Review of powertrain architectures}
The literature abounds with powertrain concepts, each with the potential to meet specific client needs~\cite{wu_powertrain_2015,yang_state---art_2016}. Perhaps the simplest architecture is using a single motor and a fixed reduction ratio between the motor and the wheels. This concept has an excellent drivability, but it introduces significant drawbacks on the vehicle design: to meet vehicular performance specifications, the motor is often oversized, and the resulting powertrain only seldom operates in its optimal efficiency region. This becomes especially problematic for heavier vehicles.
\\

A natural evolution of the single-motor fixed-ratio concept is the introduction of a multi-speed transmission. A simple concept is the manual transmission~\cite{resele_advanced_1995,lei_control_2019}. It consists of mounting gears on bearings, and selectively locking different gears to their transmission shafts to achieve different transmission ratios. Because electric motors do not need to idle, it is possible to use a manual transmission without a clutch between the motor and the transmission. To shift gears, the motor torque is first reduced to zero, then the first gear is disengaged, the motor speed is synchronized with the second gear, the second gear is engaged, and finally the motor driving torque is reapplied. Synchronizers may also help with the shaft synchronization and gear engagement~\cite{alizadeh_robust_2014,tseng_advanced_2015,mo_shift_2021,mo_dynamic_2021}. If properly performed, such a gearshift can have a low jerk level, but a torque gap is inevitable, as the shifting elements have to be engaged and disengaged when no torque is passed through them. To reduce this torque gap, we can introduce a torque gap filler in the transmission architecture~\cite{amisano_automated_2014,galvagno_analysis_2011}. This is typically a clutch placed between the motor and the transmission output shaft; it is used only during gearshifts.
\\

Alternatives to manual transmissions used in electric vehicles are dual clutch transmissions~\cite{sorniotti_analysis_2012,hong_shift_2016,gao_gear_2015} and automatic transmissions~\cite{mousavi_seamless_2015,tian_modelling_2018,fang_design_2016,roozegar_design_2017}. Conceptually, they are almost equivalent. Both consist of offering clutching and braking devices that can be modulated such that the transmission's torque can be continuously transferred between different transmission paths. The distinction resides in that dual clutch transmissions typically use a parallel shaft architecture and only require two clutches, while automatic transmissions typically use a planetary gearset architecture and require more clutching and braking devices if more than two gear ratios are to be offered. In both parallel shaft and planetary architectures, it can be interesting to replace a friction clutch by a one-way clutch~\cite{sorniotti_analysis_2012, tian_modelling_2018,ye_optimal_2017}. A one-way clutch that transfers the transmission torque in a given ratio will automatically disengage when a friction clutch of a higher gear ratio is engaged. This automatic disengagement may also simplify the gearshift control algorithm, as one fewer clutch needs to be controlled. This concept also allows to continuously transfer the torque between the two transmission paths, but with the added benefits that a one-way clutch is cheaper and more compact than a friction clutch or a brake. 
\\

Instead of using a multi-speed transmission, we could circumvent the drawbacks of having a single-motor fixed-ratio powertrain by using a plurality of motors. The different motors can be mounted on different axles, where the driving torque is shared through the road. Each motor can either power a single wheel~\cite{murata_innovation_2012} or a front or rear axle~\cite{tang_control_2010}. Alternatively, the different motors can also be mounted on the same axle~\cite{tang_dual_2013}. These architectures allow the driving torque to be continuously transferred from one motor to another, thus providing excellent drivability. This is also true for multi-motor architectures with multi-speed transmissions. For instance, a planetary gearset architecture can be configured to receive inputs from two motors~\cite{wu_robust_2018,wu_efficiency_2018,hu_efficiency_2015,wu_driving_2021}. These transmission architectures are conceptually indistinguishable from power-split transmissions used in hybrid electric vehicles~\cite{miller_hybrid_2006}. Alternatively, a parallel shaft architecture can be configured to receive inputs from two motors~\cite{sorniotti_novel_2013,liang_shifting_2018,nguyen_shifting_2020}. Such a powertrain is capable of perfectly smooth gearshifts: if the driving torque is taken exclusively from one of the two motors, the other motor transmits no torque, which allows for an easy gear change on some transmission shafts. However, a torque gap may still exist, as the torque on one motor has to be reduced to zero and the other motor may not be able to fully compensate this torque decrease. 
\\

Finally, electric powertrains can also include mechanical continuously variable transmissions~\cite{bottiglione_energy_2014}. These powertrains provide excellent drivability, as the transmission ratio can be smoothly varied. However, there is a potential concern over a reduction in the powertrain energy efficiency~\cite{sorniotti_selection_2011}.
\\

From this review we conclude that concerns over drivetrain jerk are far greater when using a single-motor multi-speed transmission concept than any other powertrain concept. However, multi-speed transmissions remain good candidates for electric vehicle designers, as they may offer the best trade-off between conflicting requirements of cost, volume, weight, and energy efficiency for a given vehicle design. This motivates the focus of this study on gearshift jerk in single-motor multi-speed transmission powertrains.
\\

In this article, the terms ``no-jerk gearshift'' and ``uninterrupted gearshift'' are used interchangeably, given that the absence of gearshift jerk implies the absence of a torque gap. In some situations a torque gap is unavoidable, so engineers are left with balancing a trade-off between minimizing the torque gap and the gearshift jerk. This is always the case with manual transmissions, as the motor torque must be reduced to zero to allow the gear change. Rapidly reducing the motor torque reduces the torque gap length, but results in a larger gearshift jerk. The intended vehicle application influences how engineers balance the conflicting drivability criteria. However in this article, we do not address this tradeoff. Instead, we focus on identifying the situations where gearshift jerk and torque gaps can be eliminated altogether. Therefore, we further focus this study on dual-clutch and automatic transmissions, since their clutch torques can be modulated to provide uninterrupted gearshifts. But as will be demonstrated in this article, clutch nonlinearities and motor saturation result in fundamental limitations to no-jerk, uninterrupted gearshifts. 

\subsection{Methodology}
The existence of fundamental limitations to no-jerk gearshifts that originate from motor saturation was hinted in~\cite{sorniotti_analysis_2012}, but such limitations were never formulated explicitly. Some studies on gearshift jerk reduction are framed around gearshift trajectory optimization~\cite{golkani_optimal_2017,ye_multi-stage_2017}. Typically, authors model a driveline, formulate a cost function that balances vehicle jerk and clutch energy dissipation, then solve a trajectory optimization problem. The main caveat with this approach is that often the optimization problem is non-convex, so a global optimum is not guaranteed. Moreover, it is hard to transfer the results to other vehicles or to slight alterations of the transmission. Other studies address the design of an optimal gearshift controller~\cite{haj-fraj_optimal_2001,walker_powertrain_2017,gao_optimal_2015,kim_gear_2017}. But similarly, this approach does not allow to generalize, as we do not know whether the vehicle jerk is a result of an imperfect controller, or an unavoidable fundamental limitation.
\\

In this article, we take a different approach where we prove that under specific circumstances, motor and actuator saturation will always cause a torque interruption during gearshifts. More precisely, we present theorems that explicitly define the fundamental limitations to no-jerk gearshifts in electric vehicles. We begin this study by defining a general driveline and vehicle model in Section~\ref{sec:model}. Then, we impose a no-jerk kinematic constraint at the vehicle level, and with the model's equations, we can solve for the resulting conditions on the transmission output shaft. This defines conditions for a no-jerk gearshift, which we present in Definition~\ref{def:no-jerk}. Subsequently, we model various transmission types and obtain specific system equations. Finally in Section~\ref{sec:thm}, we obtain the theorems by identifying the possible values of the remaining variables in the specific system equations for which a gearshift would meet the no-jerk conditions on the output shaft, and motor and actuator saturation are avoided. To facilitate the visualization of the theorems, we also present example gearshift trajectories in Section~\ref{sec:traj}. We begin Section~\ref{sec:traj} with a motor selection process for an example vehicle, so that we get realistic motor limitations. Then, we use the specific system equations to compute example trajectories where we impose the no-jerk constraints. This helps illustrate the relation between the remaining variables in the specific system equations when the no-jerk constraints are met. The objective of Section~\ref{sec:traj} is for the reader to build an intuition before being presented with the theorems in Section~\ref{sec:thm}. 
\\

In this article, we do not quantify the level of jerk or torque interruption that would result from situations where the theorems indicate that a no-jerk gearshift is impossible. This would require a specific controller design, as well as complete and well calibrated vehicle model, (see, e.g., ~\cite{holdstock_linear_2013}) which also means that the conclusions would only apply to the specific vehicle studied. Also, in the interest of generality, we work under the assumption of perfect state feedback and a perfect control of actuator forces. Formally, this translates into Assumption~\ref{assum:limits}. In practice, other system limitations will influence the dynamics of a gearshift. Nevertheless, the theorems obtained in this article help predict or identify when motor or actuator saturation is a contributor to gearshift jerk and torque interruption.
\begin{assumption} \label{assum:limits}
	Only the following two system limitations can lead to unavoidable gearshift jerk: 
	\begin{enumerate}
		\item a limit on the motor torque, which can be characterized both in terms of maximal torque $\tmax$, or maximal power $\pmax$;
		\item a limit on the torque application rate of a friction clutch, $\dTdt$. \\
	\end{enumerate}
\end{assumption}

\subsection{Contributions}
The contributions in this article include three theorems that explicitly define when an uninterrupted gearshift can be obtained for various architectures of multi-speed transmissions for electric powertrains. More specifically, Theorem~\ref{Thm:OWC} provides a necessary and sufficient condition for a no-jerk upshift when the motor operates in the power limited region, and the transmission is a dual-clutch transmission with a one-way clutch. Theorem~\ref{Thm:DCT} provides a necessary and sufficient condition for a no-jerk gearshift under the same upshift scenario, but for a dual-clutch transmission with two friction clutches. This allows to demonstrate the additional constraints on gearshift trajectories -- and therefore, the additional no-jerk limitations -- that are implied by the use of a one-way clutch. Theorem~\ref{Thm:Dwn} provides a necessary and sufficient condition for a no-jerk downshift when the motor operates in the torque limited region, also for dual-clutch transmissions. Subsequently, we show that these three theorems can be adapted to define the limitations to no-jerk gearshifts when using an automatic transmission instead.

These theorems can be used by automotive engineers to compare the potential drivability of various transmission types, thereby providing additional information during a transmission selection process for a vehicle conceptual design phase. The theorems can also be used for the design of a gearshift controller: prior to initiating a gearshift, we can now predict whether a no-jerk gearshift can be obtained under the current driving situation, or if an uncontrolled vehicle jerk is likely to occur as a result of sudden motor saturation.

\section{Vehicle, driveline, and transmission modelling} \label{sec:model}
In this section, we first present a generic vehicle and driveline model. We also introduce several transmission models, which are then embedded in the generic driveline model to compute the gearshift trajectories of Section~\ref{sec:traj}.  

\subsection{Generic vehicle and driveline models}
\begin{figure}
	\centering
	\includegraphics[width=0.7\linewidth,trim={1.2cm 22cm 7.4cm 2.5cm},clip]{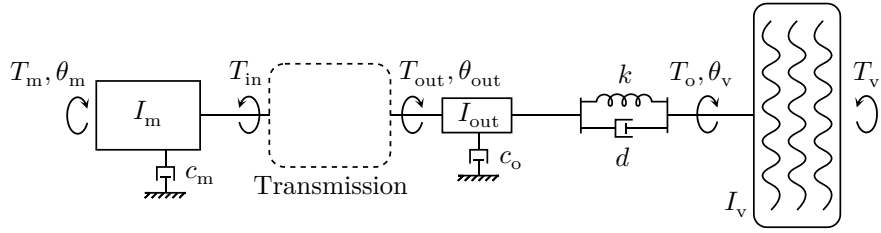}
	\caption{General driveline and vehicle model}
	\label{fig:general_model}
\end{figure}

We begin with the generalized vehicle model. We assume that the vehicle longitudinal speed $v$ follows the driving wheel speed $\tpv$ according to $v = \tpv \rw$, where $\rw$ is the wheel radius. This allows to project the vehicle mass and the vehicular forces on the wheel coordinate. The vehicle mass $m$ becomes an equivalent rotational inertia $\Iv = m \rw^2$. In this model, we consider three vehicular forces: the aerodynamic drag $F_{\mathrm{aero}}$, the tire rolling resistance $F_{\mathrm{tire}}$, and gravity $F_{\mathrm{slope}}$. These three forces become an equivalent torque $\Tv$ applied on the vehicle wheel as follows:
\begin{align}
\Tv &= \rw(F_{\mathrm{aero}} + F_{\mathrm{tire}} + F_{\mathrm{slope}}), \\
\Tv &= \rw(\frac{1}{2} \rho A_\mathrm{f} C_\mathrm{d} v^2 + m g C_\mathrm{r} \cos{\alpha} + m g \sin{\alpha}), \label{Eq:forces}
\end{align}
where $\rho$ is the air density, $A_\mathrm{f}$ is the vehicle frontal area, $C_\mathrm{d}$ is the aerodynamic drag coefficient, $C_\mathrm{r}$ is the tire rolling resistance coefficient, $g$ is gravity, and $\alpha$ is the road slope. 

We use the driveline model shown on Figure~\ref{fig:general_model}. We have three rotating bodies: the electric motor $\IIm$, the transmission output shaft $\Iout$, and the equivalent vehicle inertia $\Iv$. The three equations of motion for the general driveline and vehicle model are
\begin{align}
\IIm \tppm &= - \cm \tpm + \Tm - \Tin, \label{Eq:gen1}\\
\Iout \tppout &= - \co \tpout + \Tout - k(\tout - \tv) - d(\tpout - \tpv), \label{Eq:gen2}\\
\Iv \tppv &= - \Tv + k(\tout - \tv) + d(\tpout - \tpv),  \label{Eq:gen3}
\end{align}
where $\cm$ is the coefficient of viscous damping on the motor, and $\co$ is the coefficient of viscous damping on the transmission output. The coefficients $k$ and $d$ represent lumped driveline stiffness and damping, respectively. They cover phenomena such as driveshaft and tire flexibility and damping. 
\\

In this article, we are interested in finding gearshift trajectories that avoid vehicle jerk. Thus, we impose the kinematic constraint that $\tppv = \ar/\rw$, where $\ar$ is a prescribed and constant vehicle acceleration. By extension, we have that $\tpv = (\ar t + \vi)/\rw$, where $\vi$ is the initial vehicle speed at the beginning of the gearshift $(t=0)$. Moreover, the trajectories we study only take place for a short amount of time -- approximately 0.5\,s -- so it is fair to assume the vehicular forces $\Tv$ are constant. Given the driveline model in Equations~\ref{Eq:gen1}-\ref{Eq:gen3}, the constraints introduced on the $\tppv$, $\tpv$, and $\Tv$ variables also restrict $\tpout$ and $\tppout$. We can solve for these variables explicitly, as well as reduce the driveline model to a system of only two equations. First, the vehicular forces and acceleration are grouped into a single constant output torque $\To = \Iv \tppv + \Tv$.  Substituting $\To$ in Equation~\ref{Eq:gen3}, we get: 
\begin{equation}
	\To = k(\tout - \tv) + d(\tpout - \tpv). \label{Eq:To}
\end{equation}
Taking the time derivative of Equation~\ref{Eq:To}, we obtain:
\begin{equation}
	0 = k(\tpout - \tpv) + d(\tppout - \tppv). \label{Eq:kd}
\end{equation}
Substituting $\tppv$ and $\tpv$ with their no-jerk constraint in Equation~\ref{Eq:kd}, we obtain the following linear differential equation:
\begin{equation}
\tppout(t) + \frac{k}{d}\tpout(t) = \frac{k}{d \rw}(\ar t + \vi) + \frac{\ar}{\rw}.
\end{equation}
We can solve this equation using the initial condition that $\tpout(0) = \vi/\rw$, and we get:
\begin{equation}
\tpout(t) = \frac{\ar t + \vi}{\rw}. \label{Eq:tpoutt}
\end{equation}
We now use the prescribed trajectories on $\tpv$, $\tpout$, and $\To$ to define a no-jerk gearshift as follows.
\begin{definition} \label{def:no-jerk}
	A no-jerk gearshift is obtained if, for the duration of the gearshift,
	\begin{align}
	\tpout &= \tpv =  (\ar t + \vi)/\rw, \\
	\To &= \Iv \ar / \rw + \Tv.
	\end{align} 
\end{definition}

Finally, the general vehicle and driveline model can be reduced to only two equations of motion as follows:
\begin{align}
\IIm \tppm &= - \cm \tpm + \Tm - \Tin, \label{Eq:gen11}\\
\Iout \tppout &= - \co \tpout + \Tout - \To. \label{Eq:gen12}
\end{align}
In the next sections, we transform the general model of Equations~\ref{Eq:gen11}~and~\ref{Eq:gen12} into specific systems that depend on the transmission used, thereby replacing $\Tin$ and $\Tout$ by the relevant clutch, brake, and inertial torques. 

\subsection{Parallel shaft architecture with two frictional clutches }
\begin{figure}
	\begin{subfigure}[t]{.5\textwidth}
		\centering
		\includegraphics[width=0.8\linewidth,trim={2.0cm 23cm 12.2cm 2.0cm},clip]{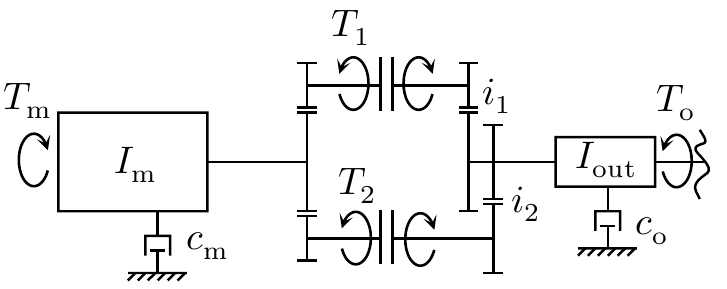}
		\caption{Two-speed transmission with parallel shaft architecture. Short name: dual-clutch transmission.}
		\label{fig:DCT}
	\end{subfigure}
	\begin{subfigure}[t]{.5\textwidth}
		\centering
		\includegraphics[width=0.8\linewidth,trim={2.0cm 23cm 12.0cm 2.3cm},clip]{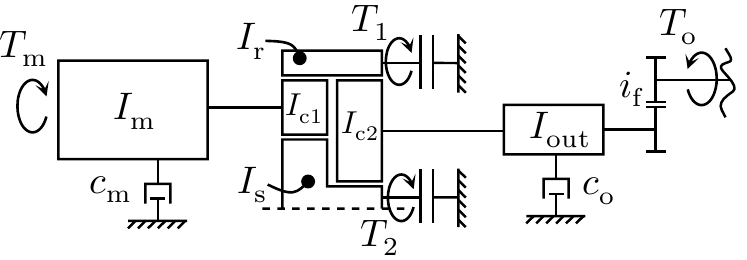}
		\caption{Two-speed transmission with a planetary gearset architecture. Short name: dual-brake transmission~\cite{mousavi_seamless_2015}.}
		\label{fig:DBT}
	\end{subfigure}
	\caption{The two powertrain architectures studied in this article.}
	\label{fig:transmissionTypes}
\end{figure}

The first system model we build is that of the architecture illustrated in Figure~\ref{fig:DCT}, when both clutches are frictional clutches. The equations of motion for this system are
\begin{align}
\IIm \tppm &= - \cm \tpm + \Tm - T_1 - T_2, \label{Eq:dct1}\\
\Iout \tppout &= - \co \tpout - \To + i_1 T_1 + i_2 T_2, \label{Eq:dct2}
\end{align}
where the clutch torques $T_1$ and $T_2$ depend on the state of the clutch -- i.e., stick or slip. We adopt the Coulomb friction model and obtain the clutch torques as follows:
\begin{align}
T_1 &= \begin{cases}
\Tm - T_2 -\IIm \tppm - \cm \tpm & \tpm = i_1 \tpout \\
\Fnone \mud \Ra n \sign(\tpm-i_1 \tpout) & \tpm \neq i_1 \tpout
\end{cases}, \label{Eq:T1}\\[5pt]
T_2 &= \begin{cases}
\Tm - T_1 -\IIm \tppm - \cm \tpm & \tpm = i_2 \tpout \\
\Fntwo \mud \Ra n \sign(\tpm-i_2 \tpout) & \tpm \neq i_2 \tpout
\end{cases},
\end{align}
where $F_\mathrm{n}$ is the linear force at the clutch plates, $\mud$ is the clutch's dynamic friction coefficient, $\Ra$ is the mean friction radius, and $n$ is the number of friction surfaces. The clutch starts to slip when the reaction torque at the interface reaches the clutch torque capacity $T_\mathrm{cap} = F_\mathrm{n} \mu_\mathrm{s} \Ra n$, where $\mu_\mathrm{s}$ is the static friction coefficient.

\subsection{Parallel shaft architecture with a one-way clutch}
If the first gear clutch is replaced by a one-way clutch in the system of Figure~\ref{fig:DCT}, the same equations of motion apply, namely Equations~\ref{Eq:dct1}~and~\ref{Eq:dct2}. But $T_1$ is a reaction torque in one direction, and it is null in the other direction: 
\begin{equation}
T_1 = \begin{cases}
\Tm - T_2 -\IIm \tppm - \cm \tpm & \tpm = i_1 \tpout \\
0 & \tpm < i_1 \tpout
\end{cases}. \label{Eq:torqueOWC}
\end{equation}
The one-way clutch also introduces the kinematic constraint $\tpm \leq i_1 \tpout$.

\subsection{Planetary gearset architecture}
In general, planetary gearsets have more degrees of freedom than parallel shaft architectures, thus they have more equations of motion. A single planetary stage can be seen as a combination of three bodies: a ring gear with inertia $\Ir$, a planet carrier ($\Ic$), and a sun gear ($\Is$). The equations of motion for each of these bodies in a single stage are
\begin{align}
\Ir \tppr &= \Tr - \Nr F, \label{Eq:Pl1}\\
\Ic \tppc &= \Tc + \Nr F + \Ns F, \label{Eq:Pl2}\\
\Is \tpps &= \Ts - \Ns F, \label{Eq:Pl3}
\end{align}
where $T_\square$ is the torque applied on either the ring (r), carrier (c), or sun (s), $\Nr$ is the radius of the ring gear, $\Ns$ is the radius of the sun gear, and $F$ is the tooth force in the gearset. There is also a kinematic constraint associated with these equations:
\begin{equation}
\Ns \tps + \Nr \tpr = (\Ns + \Nr) \tpc. \label{Eq:Plspeed}
\end{equation}
Equations~\ref{Eq:Pl1}-\ref{Eq:Plspeed} are commonly used in analyses of automatic transmissions, see \cite{bai_dynamic_2013} for instance. It is worth noting that these equations imply approximating as null the rotational inertia of the planet gears. Meanwhile, the mass of the planet gears can be considered in the equations by projecting them into the rotational inertia for the planet carrier $\Ic$.
\\

\begin{figure}
	\centering
	\includegraphics[width=0.3\linewidth,trim={1.7cm 23.5cm 16.2cm 2.5cm},clip]{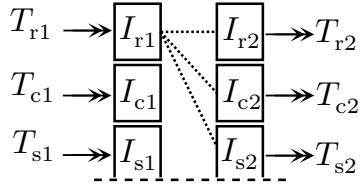}
	\caption{General representation of a double planetary gearset, where we show possible connections between the ring of the first set and elements of the second set. To operate as a transmission, a double planetary gearset must have three connections in total, which we obtain either by connecting elements together, or by grounding elements to the transmission casing. To change the transmission ratio, we simply change one of these connections. }
	\label{fig:planetary}
\end{figure}

It is also common to combine planetary gearsets in series, such as in Figure~\ref{fig:planetary}. In this case, another set of three equations of motion are added to the system (Equations~\ref{Eq:Pl1}~to~\ref{Eq:Pl3}), as well as another kinematic constraint (Equation~\ref{Eq:Plspeed}). By connecting elements, we reduce the number of degree of freedom in the system. These connections must be done carefully, as the system can become over-constrained or under-constrained. Transmission designers typically study numerous possible configurations before choosing the most suitable ones -- a process named transmission synthesis~\cite{dagci_hybrid_2018-1}, for which various methods exist ranging from using the classic lever analogy~\cite{liu_synthesis_2018}, to bond graphs~\cite{bayrak_topology_2016}. A Ravigneaux planetary gearset can be seen as a special case of a double planetary gearset with its specific set of equations. 
\\

We are now ready to build the equations for the system in Figure~\ref{fig:DBT}, which consists of a specific instance of a double planetary gearset architecture~\cite{mousavi_seamless_2015}. The inertias $\IIm$ and $\Icone$ are lumped into a single mass; we do the same for $\Iout$ and $\Ictwo$. We obtain the following equations of motion:
\begin{align}
\Ir \tppr &= T_1 - \Nrone F_1 - \Nrtwo F_2, \label{Eq:DBTEOM1}\\
\IIm \tppm &= - \cm \tpm + \Tm + \Nrone F_1 + \Nsone F_1, \\
\Iout \tppout &= - \co \tpout - \IF^{-1} \To + \Nrtwo F_2 + \Nstwo F_2, \\
\Is \tpps &= T_2 - \Nsone F_1 - \Nstwo F_2, \label{Eq:DBTEOM4}
\end{align}
and the following kinematic constraints:
\begin{align}
\Nsone \tps + \Nrone \tpr &= (\Nsone + \Nrone) \tpm, \label{Eq:DBTc1}\\
\Nstwo \tps + \Nrtwo \tpr &= (\Nstwo + \Nrtwo) \tpout. \label{Eq:DBTc2}
\end{align}
Using Equations~\ref{Eq:DBTc1}~and~\ref{Eq:DBTc2}, we can reduce the four equations of motion (Equations~\ref{Eq:DBTEOM1}~to~\ref{Eq:DBTEOM4}) into a set of two equations of motion. In order to solve this algebraic problem, researchers have assumed that elements other than the input and output shafts have negligible inertias, which simplifies the reduction process~\cite{bai_dynamic_2013}. For the system of Figure \ref{fig:DBT}, it would mean that $\Ir = 0$ and $\Is = 0$. When we make these assumptions, and introduce parameters $\beta_1 = \Nrone/ \Nsone$ and $\beta_2 = \Nrtwo/ \Nstwo$, we get a set of equations that is identical in form to that of a parallel shaft architecture, only with different coefficients. This can be observed by comparing the equations~\ref{Eq:DBT_simp1}~and~\ref{Eq:DBT_simp2} to the equations~\ref{Eq:dct1}~and~\ref{Eq:dct2}.
\begin{align}
\IIm \tppm &= - \cm \tpm + \Tm + \frac{1 + \beta_1}{\beta_1 - \beta_2} T_1 - \frac{\beta_2(1 + \beta_1)}{\beta_1 - \beta_2} T_2, \label{Eq:DBT_simp1}\\
\Iout \tppout &= - \co \tpout - \IF^{-1} \To + \frac{1 + \beta_2}{\beta_2 - \beta_1} T_1 - \frac{\beta_1(1 + \beta_2)}{\beta_2 - \beta_1} T_2. \label{Eq:DBT_simp2}
\end{align}

On the other hand, if we do not assume $\Ir = \Is = 0$, we get a different system of equations (i.e., Equations~\ref{Eq:DBTcom1}~and~\ref{Eq:DBTcom2}), more coupled this time. The constant coefficients $C_1$ to $C_8$ in these two equations are not detailed further since they are rather involved algebraic expressions that only pertain to the specific architecture of Figure~\ref{fig:DBT}. The important thing to realize is that the motor acceleration is now coupled with the transmission output.
\begin{align}
\IIm \tppm = C_1 (- \cm \tpm + \Tm) + C_2 (- \co \tpout - \IF^{-1} \To) + C_3 T_1 + C_4 T_2, \label{Eq:DBTcom1}\\
\Iout \tppout = C_5 (- \cm \tpm + \Tm) + C_6 (- \co \tpout - \IF^{-1} \To) + C_7 T_1 + C_8 T_2. \label{Eq:DBTcom2}
\end{align}

\section{Gearshift trajectories for an example vehicle} \label{sec:traj}
In Section~\ref{sec:motor}, we obtain realistic motor limitations by mimicking a motor selection process for an example vehicle. Then in Section~\ref{sec:extraj}, using the transmission models obtained in Section~\ref{sec:model}, we illustrate the resulting limitations on gearshift dynamical trajectories. The example vehicle we hypothesize is a commercial vehicle for carrying goods. With a gross vehicle mass of 8500\,kg, it would be classified as an N2 commercial vehicle in Europe, and a Class 5 medium-duty truck in North America. The vehicle and transmission parameters used are shown in Table~\ref{tbl:Params}. Note that the gross vehicle mass of 8500\,kg is used for the motor selection process, while the half-payload mass of 6500\,kg presented in Table~\ref{tbl:Params} is used for the gearshift trajectories.

\begin{table}
	\centering
	\caption{Example vehicle parameters}\label{tbl:Params}
	\begin{tabular}{C{1.5cm} L{2.5cm} | C{1.5cm} L{2.5cm} | C{1.5cm} L{2.5cm}}
		    Param.     & Value                         & Param.  & Value                             & Param.    & Value                          \\ \hline
		     $m$       & 6500\,kg                      & $\Iout$ & 0.05\,kg\,m\textsuperscript{2}    & $i_2$     & 6                              \\
		    $\rw$      & 0.3\,m                        & $\cm$   & 0.02\,Nm\,s/rad & $\Ir$     & 0.03\,kg\,m\textsuperscript{2} \\
		$A_\mathrm{f}$ & 6\,m\textsuperscript{2}       & $\co$   & 0.04\,Nm\,s/rad & $\Is$     & 0.03\,kg\,m\textsuperscript{2} \\
		$C_\mathrm{d}$ & 0.7                           & $k$     & 10\,kNm/rad & $\beta_1$ & 2                              \\
		$C_\mathrm{r}$ & 0.007                         & $d$     & 75\,Nm\,s/rad   & $\beta_2$ & 4                              \\
		    $\IIm$     & 0.3\,kg\,m\textsuperscript{2} & $i_1$   & 12                                & $\IF$     & 7.2                            \\ \hline
	\end{tabular}
\end{table}

\subsection{Motor selection} \label{sec:motor}
We set the three design specifications of Table~\ref{tbl:DesignScenarios} to define the vehicle performance requirements. The first specification consists of an extreme grade; this sets the maximal torque requirement. This is a short duration event, so we allow the powertrain to operate above its continuous capacity limit. The second specification consists of the vehicle cruising on the highway, which sets both a wheel speed requirement and a continuous power requirement. The third specification happens when the vehicle climbs a steep but reasonable grade on the highway, which sets the maximal power requirement.
\\

As shown in Figure~\ref{fig:MAP1}, if the vehicle is equipped with a single-speed transmission, the vehicle requirements can be met with a 200\,kW motor with 700\,Nm peak torque and 8000\,rpm maximum speed, using a fixed total reduction ratio of 7.5. With a two-speed transmission with ratios of 12 and 6, we can lower the peak torque requirement to 450\,Nm, while keeping the same power and speed limit requirements. In practice, this would allow vehicle designers to reduce the motor's active length~\cite{goss_design_2013}, thereby reducing the motor peak torque capacity while maintaining the same motor power capacity. We display the resulting system capacity in Figure~\ref{fig:MAP2}. The vehicle now has a 36 \% volumetrically smaller motor and a larger high-efficiency operating region. 

\begin{table}
	\centering
	\caption{Design specifications considered in the motor selection}\label{tbl:DesignScenarios}
	\begin{tabular}{l | c c l}
		Design specification     & $v$ [km/h] & $\alpha$ [\%] & Duration   \\ \hline
		1: extreme grade         & 20         & 20            & $<$ 1\,min \\
		2: highway, cruise speed & 110        & 0             & Continuous \\
		3: highway, high grade   & 90         & 5             & $<$ 1\,min \\ \hline
	\end{tabular}
\end{table}

\begin{figure}
	\begin{subfigure}[b]{.5\textwidth}
		\centering
		\includegraphics[width=0.9\linewidth,trim={0cm 0cm 0.9cm 0.5cm},clip]{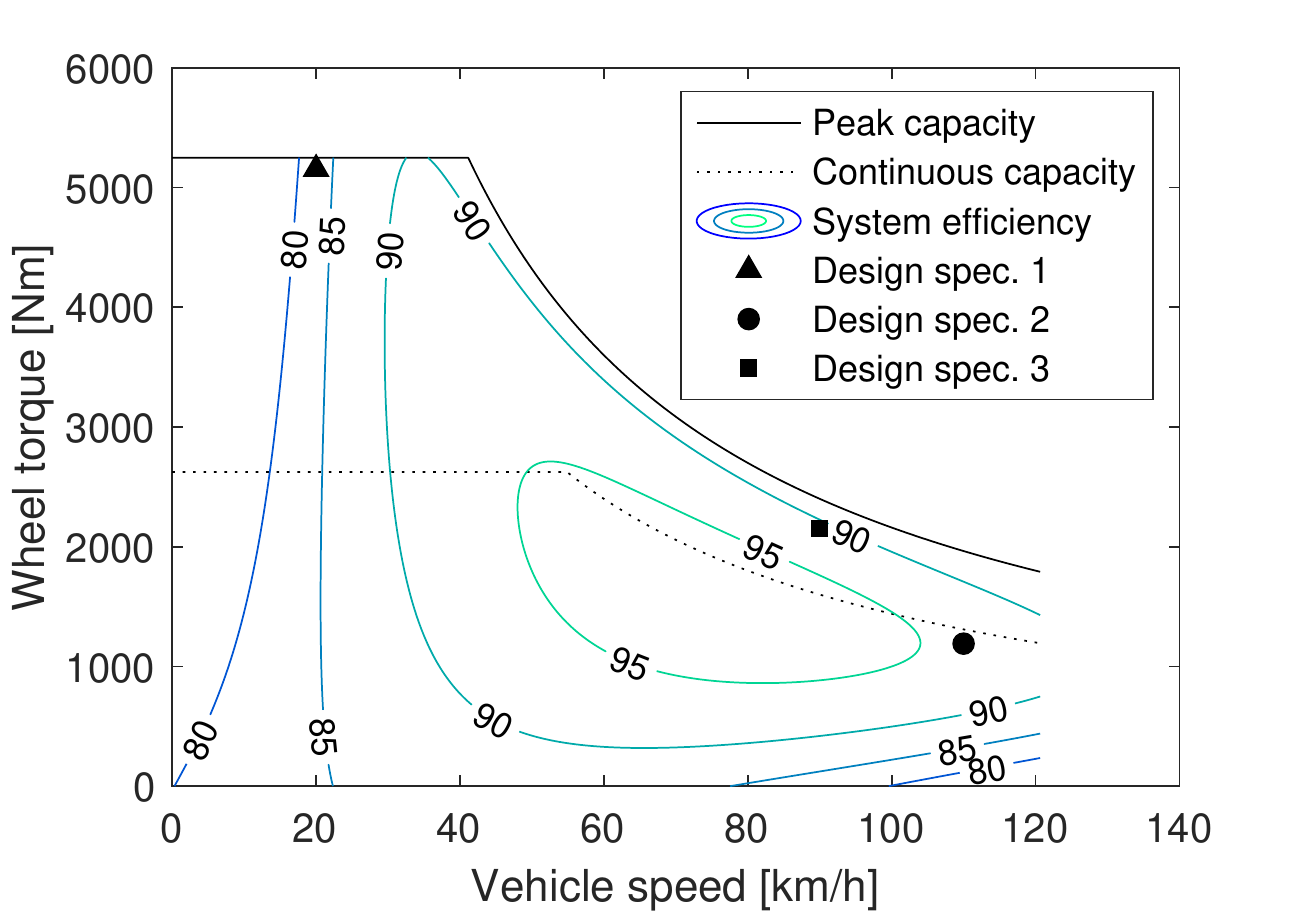}
		\caption{Single speed transmission, ratio of 7.5. Motor requirements: 200\,kW power, 700\,Nm peak torque, and 8000\,rpm maximum speed.}
		\label{fig:MAP1}
	\end{subfigure}
	\begin{subfigure}[b]{.5\textwidth}
		\centering
		\includegraphics[width=0.9\linewidth,trim={0cm 0cm 0.9cm 0.5cm},clip]{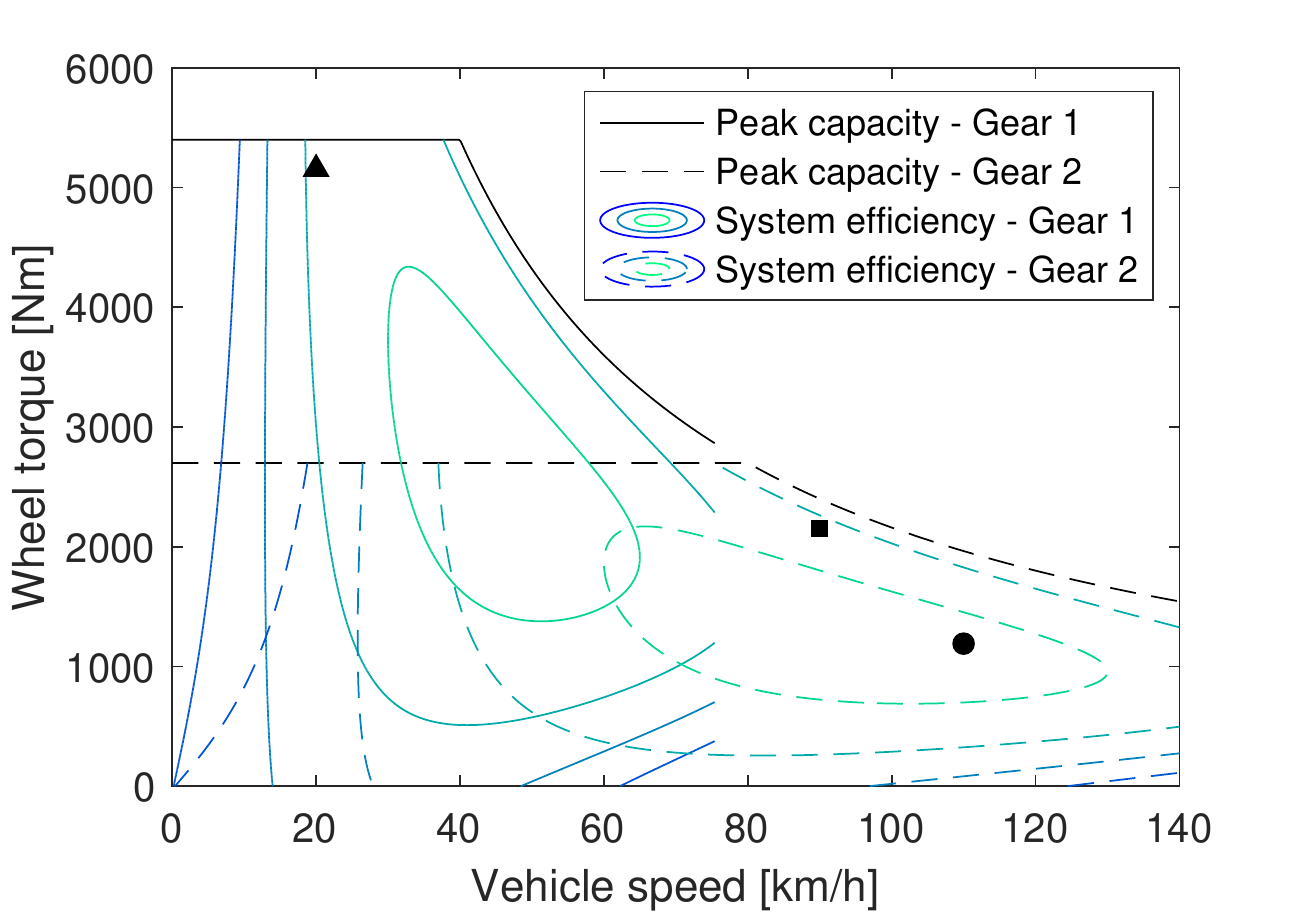}
		\caption{Two-speed transmission, ratios of 12 and 6. Motor requirements: 200\,kW power, 450\,Nm peak torque, and 8000\,rpm maximum speed.}
		\label{fig:MAP2}
	\end{subfigure}
	\caption{Vehicle capacity for (a) a single speed transmission and (b) a two-speed transmission.}
	\label{fig:MotorDesign}
\end{figure}

\subsection{Gearshift Trajectories} \label{sec:extraj}
\begin{table}
	\centering
	\caption{Gearshift scenarios}\label{tbl:GSScenarios}
	\begin{tabular}{c c c c | c c c}
		Scenario & Direction & Motor quadrant & Region         & $\vi$ [km/h] & $\ar$ [m/$s^2$] & DTD   \\ \hline
		   1     & Upshift   & Driving        & Power-limited  & 65           & $1.0$           & 80 \% \\
		   2     & Downshift & Driving        & Torque-limited & 18           & $1.0$           & 80 \% \\
		   3     & Downshift & Braking        & Torque-limited & 45           & $-1.5$          & -     \\ \hline
	\end{tabular}
\end{table}

In this section, we present five example gearshift trajectories. In all five examples, the no-jerk conditions of Definition~\ref{def:no-jerk} are imposed on the transmission output shaft. Our intention with these examples is to illustrate the implications on the input torques that result from the no-jerk conditions. More specifically, we study the three gearshift scenarios presented in Table~\ref{tbl:GSScenarios}. Scenario 1 is an upshift during vehicle acceleration, where the driver torque demand (DTD) at the beginning of the shift is 80 \% of available torque. This gearshift takes places in the power-limited region of the motor map. Scenario 2 is a downshift when the vehicle is accelerating, also with a DTD of 80 \%. The downshift takes place in the torque-limited region of the motor map. This gearshift is justified by the desire to have a greater wheel torque after the downshift. Scenario 3 is a downshift when the motor is used for regenerative braking, in the torque-limited region. 

The example trajectories were computed in \textsc{matlab}\texttrademark{} using the system models obtained in Section~\ref{sec:model}: Equations~\ref{Eq:dct1}-\ref{Eq:dct2} when simulating a dual clutch transmission, Equations~\ref{Eq:DBT_simp1}-\ref{Eq:DBT_simp2} for a dual-brake transmission where we neglect the gear inertias, and Equations~\ref{Eq:DBTcom1}-\ref{Eq:DBTcom2} for a dual-brake transmission where we include the gear inertias. First, the required $\To$ is computed considering the vehicular forces in Equation~\ref{Eq:forces} as well as the vehicle acceleration $\ar$ and initial velocity $\vi$, all prescribed by the driving scenario chosen from Table~\ref{tbl:GSScenarios}. Then, we compute the motor and clutch torques and speeds at the start of the gearshift from the model equations. Finally, we solve the system trajectories for each phase of the gearshift by imposing an arbitrary trajectory for some components, and using the model equations to solve for the trajectory of the other components. For example in a torque transfer phase, $\tpout$ follows the conditions for a no-jerk gearshift in Definition~\ref{def:no-jerk}, $\To$ is kept constant, $\tpm$ follows $\tpout$ multiplied by a given gear ratio, we impose a trajectory for $T_1$, and from the model equations, we can solve for $T_2$ and $\Tm$.
\\
\begin{figure}
	\centering
	\includegraphics[width=0.6\linewidth,trim={1.1cm 0.3cm 2.0cm 0.9cm},clip]{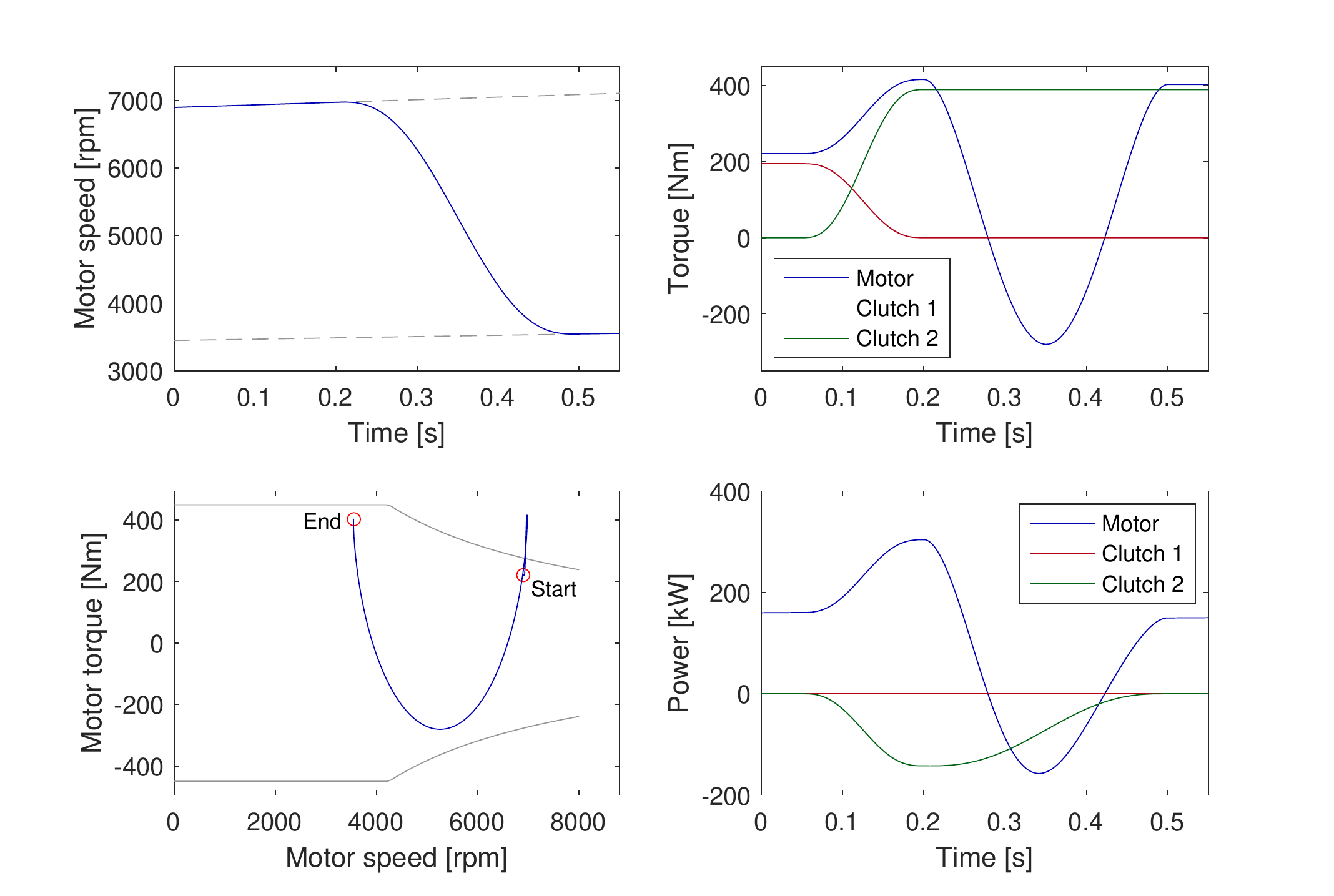}
	\caption{Example trajectory for the gearshift scenario 1 with a dual-clutch transmission, where the first clutch is a one-way clutch.}
	\label{fig:T01}
\end{figure}

The trajectory of Figure~\ref{fig:T01} is an upshift during vehicle acceleration (scenario 1) for the case where the transmission is a parallel-shaft architecture, and the first clutch is a one-way clutch. The gearshift starts at 0.05\,s and ends at 0.50\,s. It begins with a torque phase, where the transmission torque is transferred from clutch 1 to clutch 2, followed by an inertia phase, where the motor is synchronized with the gear 2 speed. During the torque phase, we gradually increase $T_2$ starting from zero torque, up to the torque required when the transmission is in gear 2, following an otherwise arbitrary trajectory. This has the effect of gradually reducing the reaction torque $T_1$ on the one-way clutch, down to zero torque at the end of the torque phase. In order to maintain a constant output torque $\To$ -- a condition for a no-jerk gearshift -- we also need to increase the motor torque $\Tm$. Furthermore, because the first clutch is a one-way clutch, this is the only possible trajectory for the system. Since the computed motor power exceeds its 200\,kW limit in this example, the trajectory is infeasible. In reality, the vehicle would inevitably experience a torque gap and vehicle jerk. This consists of the first fundamental limitation we explore in this article, see Theorem~\ref{Thm:OWC} in Section~\ref{sec:thm}, whose proof further details and formalizes the conclusions presented in this paragraph. The rest of the gearshift consists of the inertia phase, which is not as prone to motor saturation as the torque phase. $T_1$ and $T_2$ are kept constant, and the motor is synchronized with the gear 2 speed using an appropriate but arbitrary trajectory. To reduce the variation in motor power during the inertia phase -- for instance, to avoid the motor braking quadrant -- the inertia phase can simply be made longer.
\\
\begin{figure}
	\centering
	\includegraphics[width=0.6\linewidth,trim={1.1cm 0.3cm 2.0cm 0.9cm},clip]{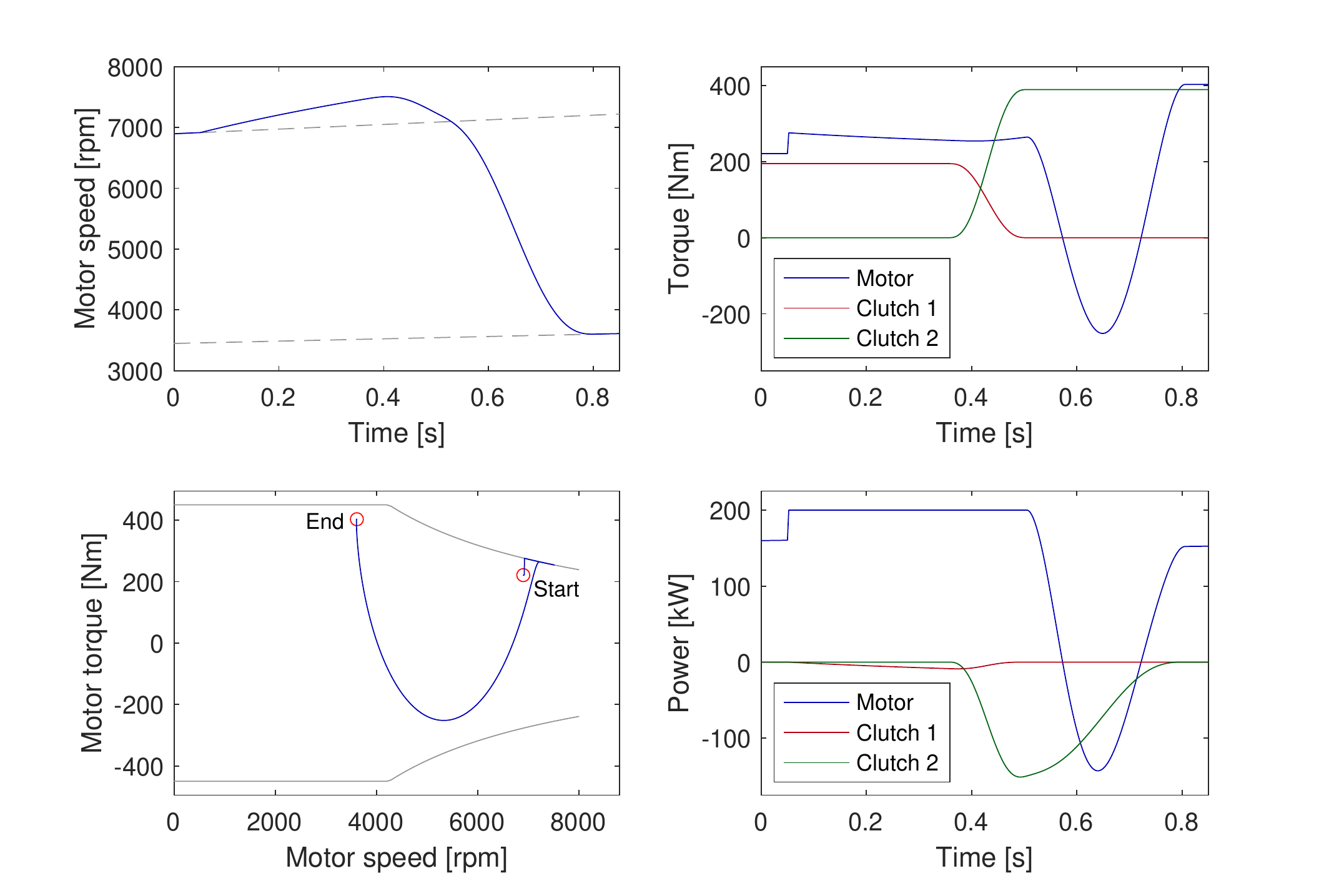}
	\caption{Example trajectory for the gearshift scenario 1 with a dual-clutch transmission, where the first clutch is a friction clutch.}
	\label{fig:T02}
\end{figure}

The trajectory of Figure~\ref{fig:T02} is also an upshift during vehicle acceleration for a parallel shaft architecture, but this time clutch 1 is a friction clutch. This introduces two new possibilities: we can increase the motor speed above the gear 1 speed, and we can modulate $T_1$ when the clutch is slipping. We take advantage of these possibilities and craft a new trajectory that results in a no-jerk gearshift. Prior to transferring the clutch torques, the motor speed is increased above the gear 1 speed, up to a prescribed value. Then begins the torque transfer, which has the effect of decreasing the motor speed. It is important that the torque transfer completes before the motor crosses the gear 1 speed, as this would result in a sign reversal on $T_1$. Recall that when a friction clutch slips, the clutch torque opposes the clutch slipping velocity, see Equation~\ref{Eq:T1} for instance. A torque reversal on $T_1$ would mean that $T_2$ has to instantaneously compensate in order to maintain the constant $\To$ condition. This is likely to violate any constraint on the clutch torque application rate. Therefore, this new trajectory also contains a fundamental limitation, which is formulated in Theorem~\ref{Thm:DCT}. 
\\
\begin{figure}
	\centering
	\includegraphics[width=0.6\linewidth,trim={1.1cm 0.3cm 2.0cm 0.9cm},clip]{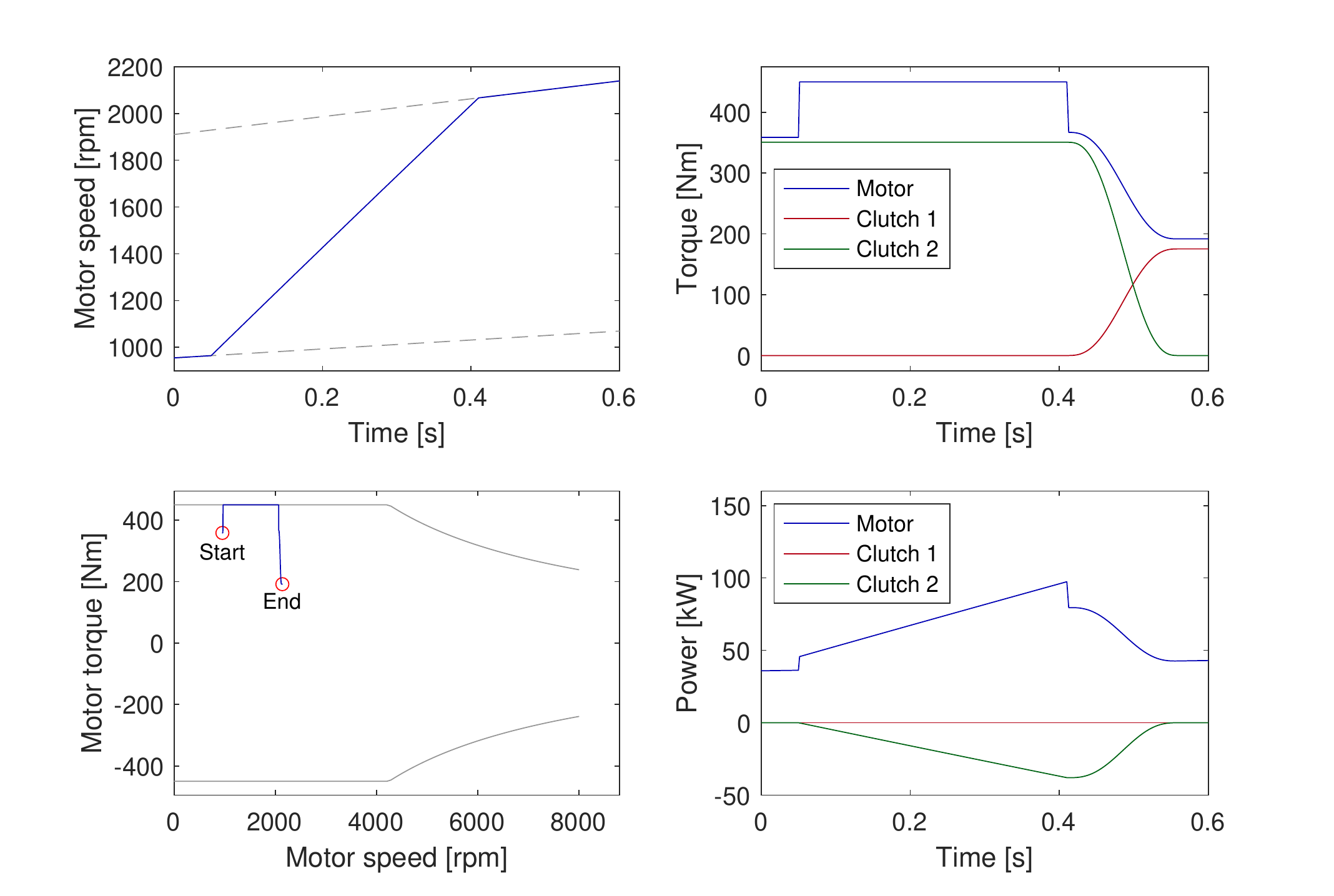}
	\caption{Example trajectory for the gearshift scenario 2 with a dual-clutch transmission. This trajectory is valid for both one-way clutch and dual-friction-clutch architectures.}
	\label{fig:T07}
\end{figure}

The trajectory of Figure~\ref{fig:T07} is a downshift during motor acceleration (scenario 2), with a parallel shaft architecture. This trajectory can be obtained for both the cases where the first clutch is a one-way clutch or a friction clutch. This time we proceed with the inertia phase before the torque phase. In the inertia phase, the motor speed is increased using maximal motor power. When the motor reaches gear 1 speed, the torque transfer begins. This time, the fundamental limitation consists of being able to synchronize the motor speed within an acceptable time while maintaining the output torque $\To$ constant. This limitation is formalized in Theorem~\ref{Thm:Dwn}. 
\\
\begin{figure}
	\centering
	\includegraphics[width=0.6\linewidth,trim={1.1cm 0.3cm 1.8cm 0.9cm},clip]{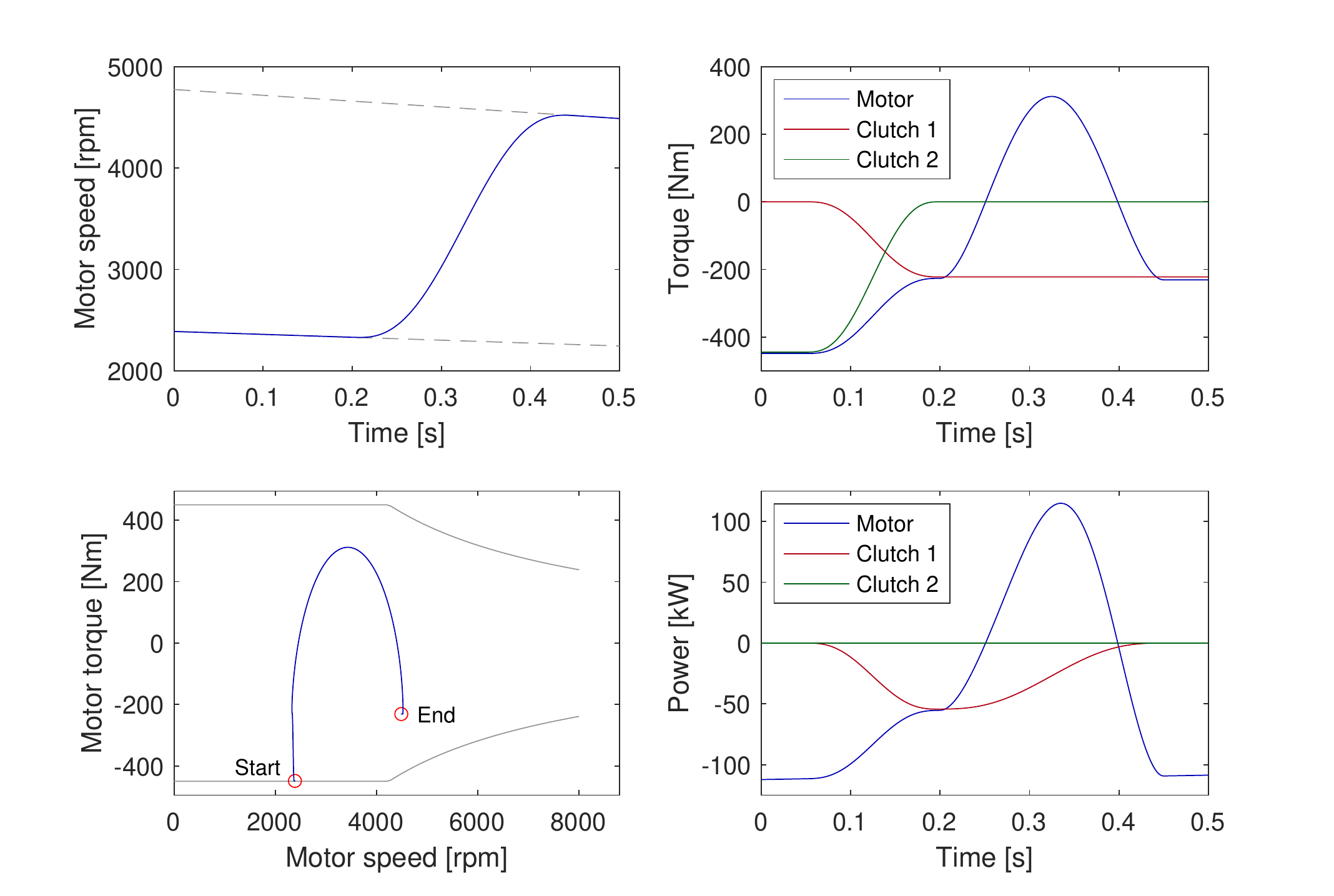}
	\caption{Example trajectory for the gearshift scenario 3 with a dual-clutch transmission. This trajectory is only possible if clutch 1 is of friction type.}
	\label{fig:T03}
\end{figure}

The trajectory of Figure~\ref{fig:T03} is a downshift during vehicle deceleration (scenario 3). More specifically, the motor operates in regenerative braking mode. We have that $\To<0$, which also means that $\Tm <0$, $T_1<0$, and $T_2<0$. Because $T_1<0$, this trajectory is only possible if clutch 1 is a friction clutch. A one-way clutch can only carry torque in one direction -- i.e., the positive direction in our case. Often, transmissions with a one-way clutch are also equipped with a locking mechanism in parallel to the one-way clutch~\cite{sorniotti_analysis_2012}. This allows for regenerative braking when the transmission operates in first gear. However, the locking mechanism typically cannot be engaged when there is a significant speed difference between the mating elements, and it cannot be modulated. Therefore, it is impossible for a dual-clutch transmission with a one-way clutch to provide uninterrupted shifting in regenerative braking mode. Figure~\ref{fig:T03} shows that for a dual friction clutch architecture, such a gearshift can be initiated even when the motor is essentially on the saturation limit. For transmissions with two friction clutches, the system limitations included in Assumption~\ref{assum:limits} will not result in gearshift jerk. 
\\
\begin{figure}
	\centering
	\includegraphics[width=0.6\linewidth,trim={1.1cm 0.3cm 1.8cm 0.9cm},clip]{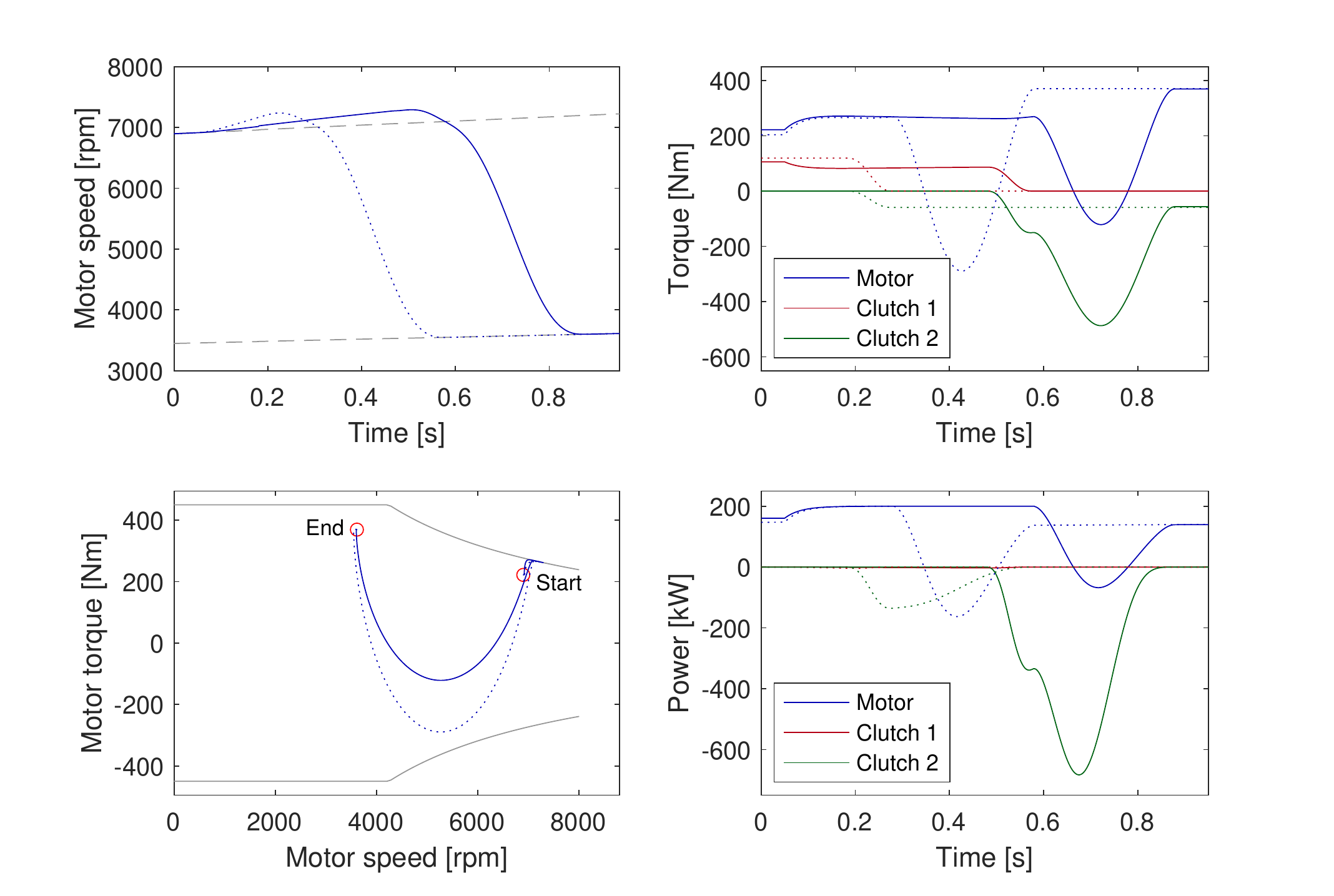}
	\caption{Example trajectories for the gearshift scenario 1 with a dual-brake transmission, namely that of Figure~\ref{fig:DBT}. The dotted lines represent the situation where the trajectory is computed with Equations~\ref{Eq:DBT_simp1}-\ref{Eq:DBT_simp2}, thereby assuming $\Ir = \Is = 0$. The solid lines represent the situation where we include the inertia of $\Ir$ and $\Is$ and compute the trajectory with Equations~\ref{Eq:DBTcom1}-\ref{Eq:DBTcom2}.}
	\label{fig:T05}
\end{figure}

In the trajectory of Figure~\ref{fig:T05}, we come back to gearshift scenario 1, but with the dual-brake transmission of Figure~\ref{fig:DBT} this time. In terms of trajectory, we follow the same strategy as that of Figure~\ref{fig:T02}, i.e., increasing the motor speed above the gear 1 speed before proceeding with the torque transfer. Computing the trajectory with Equations~\ref{Eq:DBT_simp1}-\ref{Eq:DBT_simp2}, which amounts to assuming $\Ir = \Is = 0$, we obtain the trajectory in dotted lines. But if we have that $\Ir = \Is = \IIm /10$ and compute the trajectory with Equations~\ref{Eq:DBTcom1}-\ref{Eq:DBTcom2}, we obtain the trajectory in solid lines. The introduction of a small inertia on $\Ir$ and $\Is$ has a strong influence on the resulting trajectory. First, we notice that $\Tm$ becomes coupled with $T_1$. In other words, when the motor torque is increased in the first phase of the gearshift, $T_1$ must be decreased in order to maintain the same output torque $\To$. This means $\Tm$ cannot increase  arbitrarily quickly due to the clutch torque application rate limitation. Also, we notice that the additional inertia increases the time required to attain a given motor speed during the first phase. Finally, clutch 2 takes a larger portion of the load during the inertia phase, and the motor takes a smaller portion of the load. From a design perspective, this means the maximal torque requirement on clutch 2 is higher, which has to be accounted for in the sizing of this component. In conclusion, even small sun or ring gear inertias can have a significant influence on the gearshift trajectory of a transmission with a planetary gearset architecture. Engineers should take precaution before neglecting them in the equations of motion for their systems.  

\section{Fundamental limitations to no-jerk gearshift} \label{sec:thm}
In Section~\ref{sec:thm1}, we introduce three theorems that define the fundamental limitations to no-jerk gearshift for a dual-clutch transmission. Table~\ref{tbl:SummaryCap} summarizes the association between the theorems and their pertaining combinations of gearshift scenarios and transmission types. We do not provide theorems for the gearshift scenario 3, as the conclusions follow naturally from Figure~\ref{fig:T03} and the associated discussion in Section~\ref{sec:extraj}. Then in Section~\ref{sec:thm2}, we adapt the theorems for transmissions based on planetary gearsets.

\begin{table}[h]
	\centering
	\caption{Summary of fundamental no-jerk gearshift limitations for the dual-clutch transmission}\label{tbl:SummaryCap}
	\begin{tabular}{l | c c}
		Scenario         & One-Way Clutch                   & Dual Friction Clutch             \\ \hline
		1: Upshift, driving motor  & Limited by Theorem \ref{Thm:OWC} & Limited by Theorem \ref{Thm:DCT} \\
		2: Downshift, driving motor & Limited by Theorem \ref{Thm:Dwn} & Limited by Theorem \ref{Thm:Dwn} \\
		3: Downshift, braking motor & Impossible                       & Not limited                      \\ \hline
	\end{tabular}
\end{table}

\subsection{Theorems for fundamental limitations with a dual-clutch architecture} \label{sec:thm1}
Theorem~\ref{Thm:OWC} expresses a necessary and sufficient condition for the existence of a no-jerk upshift in the power-limited region of the motor, when clutch 1 is a one-way clutch. This theorem allows one to use Equation~\ref{Eq:Pttr} and the assumptions that $\tpout(t) = (\ar t + \vi)/ \rw$ and $\tppout(t) = \ar / \rw$ to predict if a specific driving condition allows for a no-jerk gearshift; it is not needed to simulate the gearshift.
\begin{theorem} \label{Thm:OWC}
	For a dual-clutch architecture where the first gear clutch is a one-way clutch, when the motor operates in a power-limited region, a no-jerk upshift of scenario 1 can be obtained if and only if $\Pm(\ttr) \leq \pmax$, where $\Pm (\ttr)$ is the required motor power at the end of the torque transfer phase, and $\pmax$ is the motor power limit; $\Pm(\ttr)$ can be evaluated as
	\begin{equation}
	\Pm(\ttr) = i_1 \tpout(\ttr) \left( \left( i_1 \IIm + i_2^{-1} \Iout \right) \tppout(\ttr) + \left( i_1 \cm + i_2^{-1} \co \right) \tpout(\ttr) + i_2^{-1} \To \right). \label{Eq:Pttr} 
	\end{equation}
\end{theorem}

\begin{proof}
	(Necessity): The gearshift begins with a torque transfer phase that spans from $ t = 0$ to $t = \ttr$, which is defined as follows: the clutch torques are smoothly varied from $T_1(0) \neq 0$ and $T_2(0) = 0$ to $T_1(\ttr) = 0$ and $T_2(\ttr) \neq 0$, while $\tpm = i_1 \tpout$ and $\tppm = i_1 \tppout$ for all $t \in [0, \ttr]$, and $\To$, $\tpout$, and $\tppout$ follow the conditions for a no-jerk gearshift as per Definition~\ref{def:no-jerk}. The condition $\tpm = i_1 \tpout$ is necessary because the first gear clutch is a one-way clutch: Equation~\ref{Eq:torqueOWC} indicates that $ T_1 \neq 0 \rightarrow \tpm = i_1 \tpout$, so this imposes $\tpm = i_1 \tpout$ at least until $T_1 = 0$. The condition $\tppm = i_1 \tppout$ is necessary for a no-jerk gearshift: Equation~\ref{Eq:torqueOWC} indicates that if the clutch opens before $T_1 = 0$, which means we would have $\tppm < i_1 \tppout$, then the clutch torque immediately drops to 0, and Equation~\ref{Eq:dct2} indicates that for $\To$, $\tpout$, and $\tppout$ to follow the conditions for a no-jerk gearshift, then $T_2$ would have to immediately jump to a higher value, which violates the limit on clutch torque application rate. Moreover, the one-way clutch imposes the kinematic constraint that $\tpm \leq i_1 \tpout$, so if we have that $\tpm = i_1 \tpout$, we cannot have that $\tppm > i_1 \tppout$.  Thus, we showed that the torque transfer phase as defined above is necessary for a no-jerk gearshift, as any deviation from it either implies a vehicle jerk, or is physically impossible. We now find the required motor power at the end of the torque phase. With $T_1(\ttr) = 0$ in Equation~\ref{Eq:dct2}, we get that 
	\begin{equation}
	T_2(\ttr) = i_2^{-1} \left( \Iout \tppout (\ttr) + \co \tpout (\ttr) + \To \right),
	\end{equation}
	which we can substitute in Equation~\ref{Eq:dct1} to get the motor torque at $\ttr$:
	\begin{equation}
	\Tm(\ttr) = \IIm \tppm (\ttr) + \cm \tpm (\ttr) + i_2^{-1} \left( \Iout \tppout (\ttr) + \co \tpout (\ttr) + \To \right). \label{Eq:tmttr}
	\end{equation}
	We find the required motor power at $\ttr$ using $\Pm(\ttr) = \tpm(\ttr) \Tm(\ttr)$, Equation~\ref{Eq:tmttr}, and the conditions that $\tpm(\ttr) = i_1 \tpout(\ttr)$ and $\tppm(\ttr) = i_1 \tppout(\ttr)$; we obtain Equation~\ref{Eq:Pttr}. Thus a no-jerk gearshift requires a motor power $\Pm(\ttr)$ as described in~(\ref{Eq:Pttr}). Naturally, this requires that the motor be capable of producing $\Pm(\ttr)$, therefore a no-jerk gearshift implies that $\Pm(\ttr) \leq \pmax$.
	\\

	(Sufficiency): Assumption~\ref{assum:limits} indicates that if a no-jerk gearshift cannot be obtained, it is because either the motor saturates, or the clutch saturates. Further assuming that $\ttr$ is large enough such that the clutch application rate does not saturate, if we cannot obtain a no-jerk gearshift, it is because the motor saturates. By showing that $\Pm(\ttr)$ is the maximal required motor power during the gearshift, we can show that if the motor saturates, we must have that $\Pm(\ttr) > \pmax$. To do this, let us rearrange Equations~\ref{Eq:dct1}~and~\ref{Eq:dct2} to eliminate $T_2$ and isolate $\Tm$ to get the motor power as follows,
	\begin{equation}
	\Pm = \tpm \left( \IIm \tppm + \cm \tpm + T_1 \left( 1-\frac{i_1}{i_2} \right) + i_2^{-1} \left( \Iout \tppout + \co \tpout + \To \right) \right). \label{Eq:Pmax}
	\end{equation}
	Assuming that the motor begins to synchronize with the second gear speed the instant the torque transfer is complete at $\ttr$, we have that the maximal $\tpm(t)$ is at $t = \ttr$. Furthermore, every term on the right-hand side of Equation~\ref{Eq:Pmax} is maximized at $\ttr$. In particular, we notice that since $i_1>i_2$, $T_1(\ttr) = 0$ maximizes $\Pm$, as $T_1 \geq 0$. Therefore, if we cannot obtain a no-jerk gearshift, then $\Pm(\ttr)>\pmax$.
\end{proof}

In Theorem~\ref{Thm:DCT}, we express a necessary and sufficient condition such that we can have a no-jerk upshift (gearshift scenario 1) in the power-limited region of the motor, when clutch 1 is a friction clutch. We follow the strategy detailed in Figure~\ref{fig:Trajectory}. To give the reader a visual appreciation of the limitations, we also simulated several instances of this strategy, which are displayed in Figure~\ref{fig:LIM}. In practice, one can use Theorem~\ref{Thm:DCT} to predict whether a specific driving condition allows a no-jerk gearshift. Because Equation~\ref{Eq:Th2Tm} makes it such that the equations of motion for the system become nonlinear, there may not be a convenient closed-form solution to facilitate the process of finding a sufficient $\DM$. Perhaps the best way to do so is to simulate the first phase of the gearshift up to a given $\DM$, and then validate if this $\DM$ is sufficient to allow a complete torque transfer before $\DS = 0$. If the given $\DM$ is not sufficient, then the process can be repeated for a higher $\DM$, until no higher $\DM$ can be reached, at which point we can conclude a no-jerk trajectory is infeasible.

\begin{theorem} \label{Thm:DCT}
	Consider a dual clutch architecture with two friction clutches as shown in Figure~\ref{fig:DCT}. Referring to Figure~\ref{fig:Trajectory}, suppose that at $t=0$ the motor speed is synchronized with gear 1's speed and that it is desired to initiate a gearshift to be completed at time $t_2$. Let $\ttr>0$ be the set torque transfer duration in the gearshift and $t_1$ the time at which the torque transfer is initiated. When the motor operates in a power-limited region, a no-jerk upshift of scenario 1 can be obtained if and only if
	\begin{align}
	\exists \ &\DM:= \tpm(t_1) - i_1 \tpout(t_1) \geq 0 \label{Eq:Th2E1},\\
	\mathrm{such\ that} \ &\DS:= \tpm(t_1 + \ttr) - i_1 \tpout(t_1 + \ttr) \geq 0,  \\
	&\tpm(t_1) \leq \tpmax, \\
	\mathrm{where} \ &\Tm(t) = \pmax / \tpm(t), \quad 0 \leq t \leq t_1+\ttr, \label{Eq:Th2Tm} \\
	&T_1(t) = 0, \quad t_1 + \ttr \leq t \leq t_2,  \label{Eq:Th2T1} \\
	&T_2(t) = 0, \quad 0 \leq t \leq t_1. \label{Eq:Th2T2}
	\end{align}
\end{theorem}

\begin{proof}
	(Necessity): The gearshift begins with a speed phase ($0 \leq t \leq t_1$), where the motor speed is increased above $i_1 \tpout$, with $T_2 = 0$ and $\Tm = \pmax /  \tpm$. Then the gearshift continues with a torque transfer phase ($t_1 \leq t \leq t_1 + \ttr$), where the clutch torques are smoothly varied from $T_1(t_1) \neq 0$ and $T_2(t_1) = 0$ to $T_1(t_1 + \ttr) = 0$ and $T_2(t_1 + \ttr) \neq 0$, while $\Tm = \pmax / \tpm$. Finally, the gearshift ends with an inertia phase ($t_1 + \ttr \leq t \leq t_2$), where the motor speed is brought down to $i_2 \tpout$, while $T_1 = 0$. During all three phases, the conditions for a no-jerk gearshift in Definition~\ref{def:no-jerk} can be maintained by modulating $T_1$ and $T_2$ as per Equation~\ref{Eq:dct2}. Motor saturation can be avoided as the motor torque is set to $\Tm = \pmax / \tpm$ for the first two phases, and the motor synchronization of the inertia phase requires that $\Tm < \pmax / \tpm$. The torque application rate on both clutches can be maintained within limits during the speed phase, as $T_2 = 0$ and the variations on $\To$, $\tpout$, and $\tppout$ are small enough such that $\dTonedt$ is within limits. The same argument can be made for the inertia phase, with $T_1 = 0$ this time. During the torque transfer phase, appropriate clutch torque trajectories must be chosen such that limits on $\dTdt$ are respected. Moreover, a torque reversal on clutch 1 must be avoided.
	
	In effect, the clutch torque $T_1$ is necessarily positive when the motor is driving the vehicle through the first gear and clutch 1 sticks. When clutch 1 slips, the direction of $T_1$ is dependent on the slip direction, as described in Equation~\ref{Eq:T1}. If $\tpm < i_1 \tpout$, then the torque $T_1$ suddenly reverses direction and becomes negative. In order to maintain the conditions for a no-jerk gearshift, Equation~\ref{Eq:dct2} indicates that the torque reversal on $T_1$ must be instantaneously compensated by an increase in $T_2$, which necessarily violates any application rate limitation. Consequently, we must have that $\tpm \geq i_1 \tpout$ as long as $T_1 \neq 0$. Once the torque transfer phase begins, we have $\tppm(t) < \tppm(t_1)$, as $i_1 > i_2$, which can be seen from Equations~\ref{Eq:dct1}~and~\ref{Eq:dct2}. Therefore, the torque transfer needs to start at a sufficiently high $\DM$ such that $\DS \geq 0$. Moreover, $\DM$ must correspond to a motor speed that is within its limit $\tpmax$. However, nothing guarantees the system can reach such a $\DM$. If it fails to reach a satisfactory $\DM$, then one of the conditions for a no-jerk gearshift will not be satisfied due to system limitations -- there will be jerk. This proves the necessity part by contraposition.
	\\
		
	(Sufficiency): By construction, the actuation strategy described in Equations~\ref{Eq:Th2Tm}~to~\ref{Eq:Th2T2} is sufficient for a no-jerk gearshift, given that Assumption~\ref{assum:limits} holds.
	
	This actuation strategy is not unique, but any deviation would only make it harder to achieve a sufficiently high $\DM$. Consequently, if a no-jerk gearshift can be obtained using such a deviation from Equations~\ref{Eq:Th2Tm}~to~\ref{Eq:Th2T2}, it can also be obtained using the actuation strategy described in these equations. Therefore, if a no-jerk gearshift cannot be obtained, there is no $\DM$ such that $\DS \geq 0$ when following Equations~\ref{Eq:Th2Tm}~to~\ref{Eq:Th2T2}.
\end{proof}

\begin{figure}
	\centering
	\includegraphics[width=0.4\linewidth,trim={2.6cm 21.8cm 11cm 2.2cm},clip]{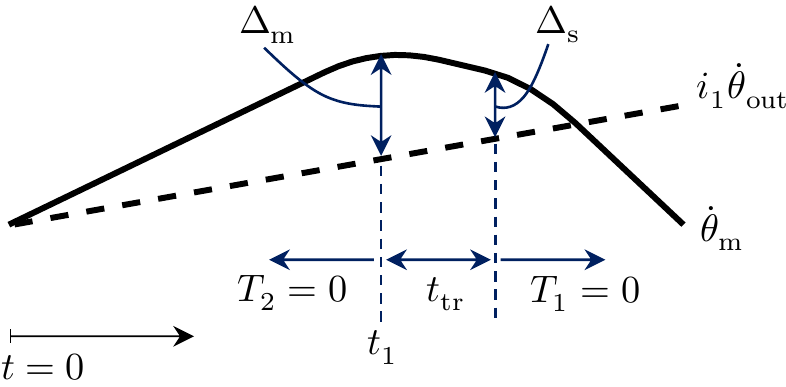}
	\caption{Strategy for a power-on upshift with a dual-friction clutch transmission. We first increase the motor speed to an increment $\DM$ above the gear 1 speed. Once $\DM$ is reached, we begin the torque transfer, which takes a set time $\ttr$. At the end of the torque transfer, the motor will have decelerated to a different speed whose increment over gear 1 speed we define as $\DS$. In order to avoid torque reversal at clutch 1, which would imply an unavoidable jerk, we must have that $\DS \geq 0$.}
	\label{fig:Trajectory}
\end{figure}

\begin{figure}
	\begin{subfigure}[t]{.5\textwidth}
		\centering
		\includegraphics[width=0.8\linewidth,trim={0cm 0cm 1cm 0.2cm},clip]{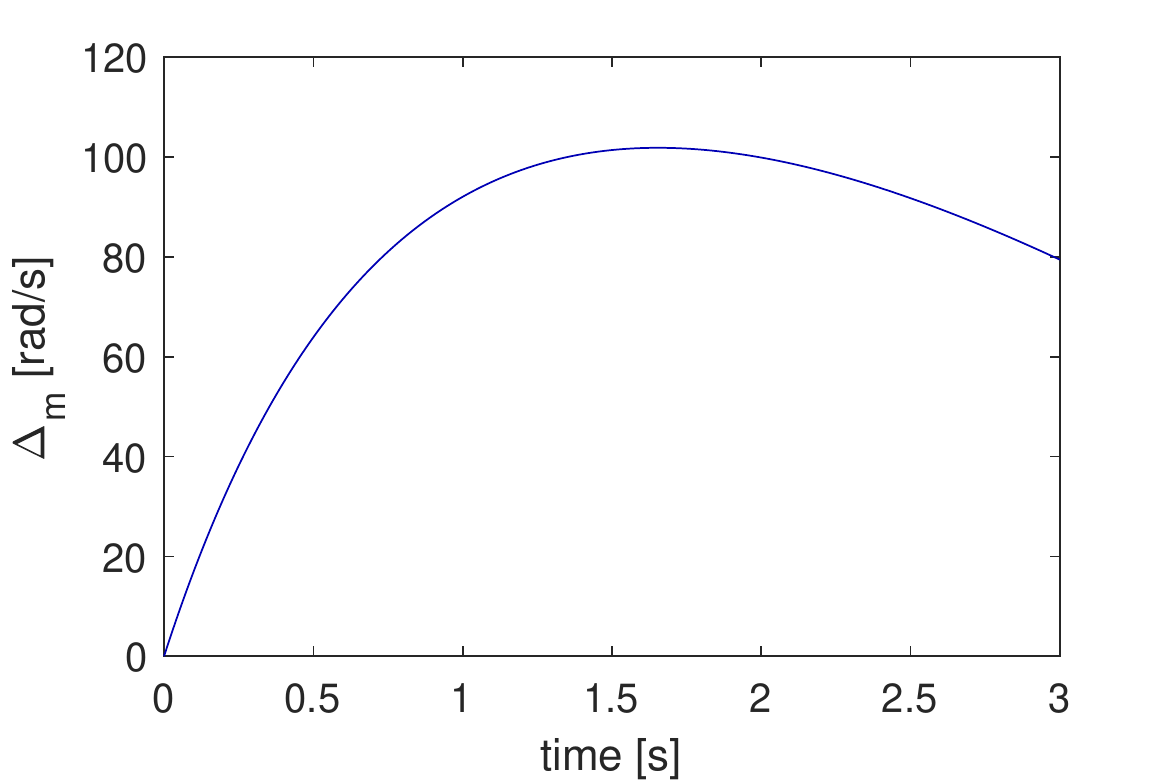}
		\caption{Even when using the maximum available motor power, not all $\DM$ can be reached.}
		\label{fig:LIM1}
	\end{subfigure}
	\begin{subfigure}[t]{.5\textwidth}
		\centering
		\includegraphics[width=0.8\linewidth,trim={0cm 0cm 1cm 0.4cm},clip]{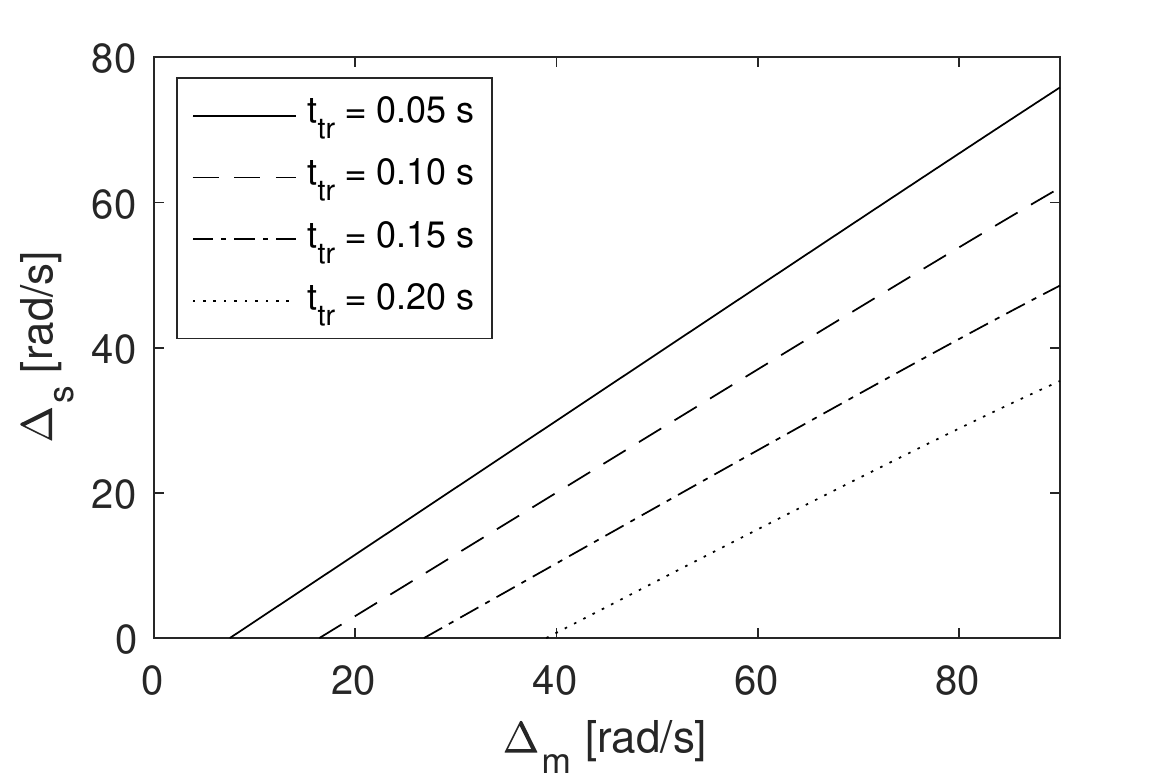}
		\caption{The resulting $\DS$ after the torque transfer is a function of the starting $\DM$ and the transfer time $\ttr$. }
		\label{fig:LIM2}
	\end{subfigure}
	\caption{Example of limitations for a power-on upshift with the example vehicle and a dual-friction-clutch transmission, when the gearshift is initiated at $\vi$ = 65\,km/h, and $\ar$ = 1.0\,m/s\textsuperscript{2}}
	\label{fig:LIM}
\end{figure}

Comparing Theorems~\ref{Thm:OWC}~and~\ref{Thm:DCT}, we conclude that transmissions with a one-way clutch are more prone to motor saturation. In effect, the constraint on the maximal motor power imposed in Theorem~\ref{Thm:OWC} does not exits in Theorem~\ref{Thm:DCT}, and as a result, transmissions with two friction clutches have a wider set of possible no-jerk gearshift trajectories. For example, both gearshift trajectories of Figure~\ref{fig:T01} (one-way clutch) and Figure~\ref{fig:T02} (two friction clutches) are initiated under the same driving scenario, namely an upshift at 80\% DTD in the motor's power-limited region. The trajectory of Figure~\ref{fig:T01} results in motor saturation and Theorem~\ref{Thm:OWC} would indicate that it is inevitable, meanwhile the trajectory of Figure~\ref{fig:T02} completes without motor saturation and a no-jerk gearshift is obtained.
\\

In Theorem~\ref{Thm:Dwn}, we express a necessary and sufficient condition such that we can have a no-jerk downshift (gearshift scenario 2) in the torque limited region of the motor. This limitation is the same for both when the first clutch is a one-way clutch or a friction clutch. The limitation is illustrated in Figure~\ref{fig:LIMdownshift}. We simulated several instances of gearshift scenarios 2, where we varied the initial vehicle acceleration $\ar$. We recorded the required time for the motor to synchronize with the first gear speed, which we label $\ts$. When $\ar$ is too high, either $\ts$ is impractically long, or the motor never synchronizes with the first gear speed.
\begin{theorem} \label{Thm:Dwn}
	For a dual-clutch architecture with two friction clutches, when the motor operates in a torque-limited region, assuming a constant $T_2$ during the inertia phase, a no-jerk power-on downshift of scenario 2 can be obtained if and only if 
	\begin{equation}
	\exists \ \ts >0 \quad \mathrm{s.t.} \quad 0 = - i_1 \frac{\vi + \ar \ts}{\rw} + \exp \left(-\frac{\cm}{\IIm}\ts \right) i_2 \frac{\vi}{\rw} + \frac{\tmax - T_2}{\cm} \left[ 1- \exp \left(-\frac{\cm}{\IIm}\ts \right) \right]. \label{Eq:thm3}
	\end{equation}
\end{theorem}

\begin{proof}
	(Necessity): The gearshift begins with an inertia phase $(0 \leq t \leq \ts)$, where $\tpm$ is accelerated from gear 2 synchronization speed to gear 1 synchronization speed, while $\Tm = \tmax$ and $T_1 = 0$. The gearshift ends with a torque phase $(\ts \leq t \leq \ts + \ttr)$, where the transmission torque is transferred from clutch 2 to clutch 1. During both phases, the conditions for a no-jerk gearshift in Definition~\ref{def:no-jerk} can be maintained by modulating $T_1$ and $T_2$ as per Equation~\ref{Eq:dct2}. Motor saturation can be avoided by restricting the motor power to $\Tm \leq \tmax$. Clutch torque application rate saturation can be avoided by using an appropriate torque transfer trajectory during the torque phase. However, when using these restrictions together, nothing guarantees that the motor will eventually synchronize with gear 1 speed at some time $\ts$. We now find an equation that describes the synchronization time $\ts$ when the no-jerk gearshift conditions are maintained. With $\Tm = \tmax$ and $T_1 = 0$ in Equation~\ref{Eq:dct1}, we get the motor acceleration
	\begin{equation}
	\tppm = \IIm^{-1} \left( - \cm \tpm + \tmax - T_2 \right). \label{Eq:thm3tppm}
	\end{equation} 
	Due to the reasonable assumption that $T_2$ is constant during the inertia phase, Equation~\ref{Eq:thm3tppm} is a linear ordinary differential equation. We solve Equation~\ref{Eq:thm3tppm} to obtain the evolution of $\tpm(t)$. We impose the initial condition that $\tpm(0) = i_2 \tpout(0)$, and we get
	\begin{equation}
	\tpm(t) = \exp \left(-\frac{\cm}{\IIm}t \right) i_2 \tpout(0) + \frac{\tmax - T_2}{\cm} \left[ 1- \exp \left(-\frac{\cm}{\IIm}t \right) \right]. \label{Eq:thm3tpm}
	\end{equation}
	When the motor synchronizes with the gear 1 speed, we have that $\tpm = i_1 \tpout$. Substituting this condition into Equation~\ref{Eq:thm3tpm}, and using the fact that $\tpout(t) = (\vi + \ar t)/ \rw$, we obtain Equation~\ref{Eq:thm3}. For the motor to synchronize with gear 1 speed, Equation~\ref{Eq:thm3} must have a solution. Therefore, if Equation~\ref{Eq:thm3} has no solution for $\ts$, then a no-jerk gearshift cannot be obtained, which proves the necessity part by contraposition.
	\\ 
	
	(Sufficiency): By construction, the actuation strategy described in the first paragraph of this proof is sufficient for a no-jerk gearshift, given that Assumption~\ref{assum:limits} holds. 
	
	This actuation strategy is not unique. We could have $\Tm < \tmax$, but it would only make it harder to synchronize the motor with gear 1 speed, which can be seen from Equation~\ref{Eq:thm3tppm}. If clutch 1 is a one-way clutch, then we must have $T_1 = 0$ until the motor synchronizes with gear 1 speed, as per Equation~\ref{Eq:torqueOWC}. If clutch 1 is a friction clutch, then it is possible to activate $T_1$ before $\ts$. But from Equation~\ref{Eq:T1}, we notice that we would have $T_1 < 0$, so Equation~\ref{Eq:dct2} dictates that $T_2$ increases by $|T_1|(i_1/i_2)$ to respect the no-jerk conditions, and since $i_1>i_2$, $\tppm$ would again be smaller than if $T_1 = 0$. Therefore, if a no-jerk gearshift can be completed with a different actuation strategy than the one presented in the necessity part of the proof, it can also be completed with this strategy. As a result, the existence of a solution to Equation~\ref{Eq:thm3} is a sufficient condition for the possibility of a no-jerk gearshift.
\end{proof}

\begin{figure}
	\centering
	\includegraphics[width=0.40\linewidth,trim={0cm 0cm 1cm 0.4cm},clip]{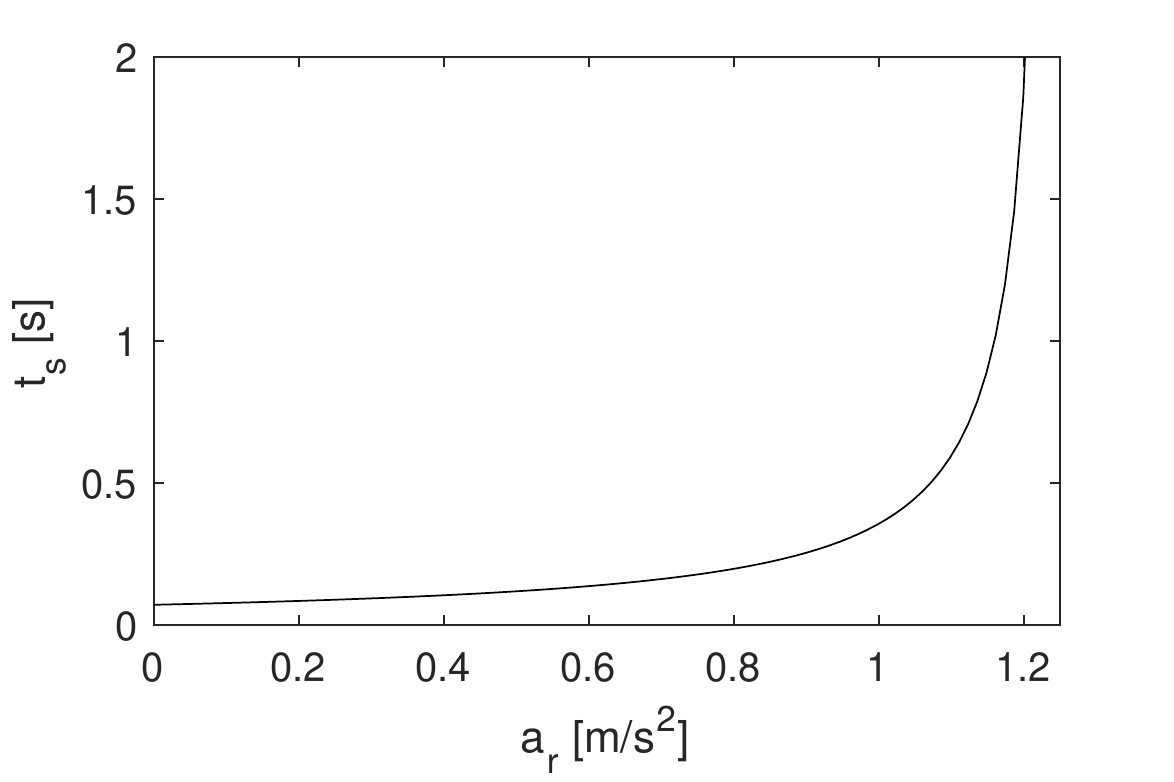}
	\caption{Example of limitations for a power-on downshift with the example vehicle, when the gearshift is initiated at $\vi$ = 18\,km/h, and $\ar$ is varied. When the desired vehicle acceleration $\ar$ is increased, the time $\ts$ required for the inertia phase to complete also increases.}
	\label{fig:LIMdownshift}
\end{figure}

\subsection{Theorem adaptations for planetary gearset architectures} \label{sec:thm2}
In this section, we adapt Theorems~\ref{Thm:OWC}~to~\ref{Thm:Dwn} to planetary gearset architectures described by the general Equations~\ref{Eq:DBTcom1}~and~\ref{Eq:DBTcom2}. 

\subsubsection{Theorem 1 adaptation}
With the new system equations, the motor power at the end of the torque transfer phase -- originally described by Equation~\ref{Eq:Pttr} -- now becomes
\begin{multline}
\Pm(\ttr) = i_1 \tpout (\ttr) \Bigg( \cm i_1 \tpout (\ttr) + \bigg( C_1 - \frac{C_4 C_5}{C_8} \bigg)^{-1} \bigg[ \bigg( i_1 \IIm - \frac{C_4}{C_8} \Iout \bigg) \tppout (\ttr) \\
+ \bigg( C_2 - \frac{C_4 C_6}{C_8} \bigg) \big( \co \tpout (\ttr) + \IF^{-1} \To \big) \bigg] \Bigg).
\end{multline}
The proof for the necessity condition remains valid, as it it based on the limitations imposed by the one-way clutch, and these limitations still apply. The proof for the sufficiency condition requires to demonstrate that $\Pm(\ttr)$ is the maximal motor power required during the gearshift, which now ultimately depends on the coefficients $C_1$ to $C_8$, as can be seen by adapting Equation~\ref{Eq:Pmax} into a new expression for the motor power:
\begin{multline}
\Pm = \tpm \Bigg( \cm \tpm + \bigg( C_1 - \frac{C_4 C_5}{C_8} \bigg)^{-1} \bigg[ \IIm \tppm - \frac{C_4}{C_8} \Iout \tppout  \\
+ \bigg( C_2 - \frac{C_4 C_6}{C_8} \bigg) \big( \co \tpout + \IF^{-1} \To \big) - \bigg( C_3 - \frac{C_4 C_7}{C_8} \bigg) T_1 \bigg] \Bigg).
\end{multline}
In particular, for $T_1 = 0$ to maximize $\Pm$, we must have $-(C_3 - C_4C_7/C_8)(C_1 - C_4C_5/C_8)^{-1}<0$. For the architecture in Figure~\ref{fig:DBT}, this expression reduces to $-(\beta_1+1)/\beta_1$, so that indeed $T_1 = 0$ maximizes $\Pm$ since $-(\beta_1+1)/\beta_1<0$. Assuming that the other variables -- i.e., $\tpm$, $\tppm$, $\tpout$, $\tppout$, and $\To$ -- remain approximately constant during the torque phase, we have that $\Pm(\ttr)$ is the maximal motor power required during the gearshift for the case of the architecture in Figure~\ref{fig:DBT}. 

\subsubsection{Theorem 2 adaptation} \label{sec:adaptTh2}
The existence of a sufficient $\DM$ such that $\DS \geq 0$ remains a necessary and sufficient condition for the possibility of a no-jerk gearshift. The only difference is that we cannot have $\Tm$ jumping to $\pmax / \tpm$ at $t=0$, as prescribed in Equation~\ref{Eq:Th2Tm}. In effect, $\Tm$ now appears in the second equation of motion of the system, i.e., Equation~\ref{Eq:DBTcom2}, so to maintain $\To$, $\tpm$, and $\tppm$ such that the no-jerk conditions in Definition~\ref{def:no-jerk} are met, a sudden increase in $\Tm$ implies a sudden increase in $T_1$ (as we impose $T_2(0)=0$), which necessarily violates any torque application rate limitation. The effect of an increase in $\Tm$ on $T_1$ can be seen in Figure~\ref{fig:T05}. So Theorem~\ref{Thm:DCT} remains valid, but we must ramp up $\Tm$ such that $\dTdt$ is within limits. 

\subsubsection{Theorem 3 adaptation}
The proof for Theorem~\ref{Thm:Dwn} remains valid, but the necessary and sufficient condition described by Equation~\ref{Eq:thm3} must be adapted to the new system equations. First, we get $\tpm$ during the inertia phase by imposing $T_1 = 0$ in Equations~\ref{Eq:DBTcom1}~and~\ref{Eq:DBTcom2}; we get
\begin{align}
\tppm &=  \IIm^{-1} \left( -\gamma \cm \tpm + \gamma \Tm - \tau \right),     \\
\mathrm{where} \quad \gamma &= \left( C_1 -\frac{C_4 C_5}{C_8} \right), \\
\tau &= \left(C_2 - \frac{C_4 C_6}{C_8} \right) \left( \co \tpout + \IF^{-1} \To \right) + \frac{C_4}{C_8} \Iout \tppout.
\end{align}
Following the argument in Section~\ref{sec:adaptTh2}, we cannot have $\Tm$ jump to $\tmax$ at $t=0$, as this would imply a jump in $T_2$, which would violate any $\dTtwodt$ limitation. Therefore, $\Tm$ should be increased gradually from $\Tm(0)$ to $\tmax$ at the beginning of the gearshift. For this proof adaptation however, we neglect the effect of gradually ramping $\Tm$ on the synchronization time $\ts$ -- we assume the ramp is completed very quickly, so $\Tm \approx \tmax$. Further, we assume that $\tpout(t) = \tpout(0)$ and $\To(t) = \To(0)$ for the duration of the inertia phase, which we also assumed in Theorem~\ref{Thm:Dwn} by imposing $T_2$ constant for the duration of the inertia phase. We now have a linear ordinary differential equation of the same form as in Theorem~\ref{Thm:Dwn}, which we also solve imposing $\tpm(0) = i_2 \tpout(0)$. We obtain an adapted condition for the possibility of a no-jerk gearshift as follows:
\begin{equation}
\exists \ \ts >0 \quad \mathrm{s.t.} \quad 0 = - i_1 \frac{\vi + \ar \ts}{\rw} + \exp \left(-\frac{\gamma \cm}{\IIm}\ts \right) i_2 \frac{\vi}{\rw} + \frac{\gamma \tmax - \tau}{\gamma \cm} \left[ 1- \exp \left(-\frac{\gamma \cm}{\IIm}\ts \right) \right]. \label{Eq:thm3adapt}
\end{equation}

\subsection{Theorem adaptation for other vehicle models and kinematic conditions} \label{sec:thmadaptkin}
The theorems presented in this article are based on the general driveline model described in Section~\ref{sec:model}, which does not include tire slip and other nonlinear behaviors likely to occur on real-world vehicles. Should one desire to adapt the theorems to other vehicle models, they would need to solve for the required transmission output torque $\To(t)$ and velocity $\tpout(t)$ using the new equations for the vehicle model and the no-jerk constraints they wish to impose. Effectively, this would result in different conditions for a no-jerk gearshift than what we present in Definition~\ref{def:no-jerk}. With these new conditions for $\To(t)$ and $\tpout(t)$, the theorems can be adapted accordingly, perhaps with additional mathematical complexity as more terms may become time-dependent.
\\

Similarly, one could adapt the theorems to other kinematic conditions than the no-jerk condition we imposed $(\tpppv = 0)$, such as small and constant vehicle jerk level for instance. This also requires solving for new conditions on $\To(t)$ and $\tpout(t)$, and adapting the theorems accordingly. The same remark holds: this may increase the mathematical complexity as more terms may become time-dependent.

\section{Conclusion}
In this article, we presented theorems that describe the fundamental limitations to no-jerk gearshift in the presence of motor and clutch saturation. We showed that transmissions with a one-way clutch have stronger limitations than their friction clutch counterparts. This means that a one-way clutch transmission will fail to provide an uninterrupted gearshift under a wider set of driving conditions. We also showed that transmissions with a planetary gearset architecture have different dynamics than transmissions with parallel shaft architectures, which requires slight adaptations of the theorems for fundamental limitations of no-jerk gearshifts. 
\\

This work has important implications for automotive engineers. The theorems are tools to quickly evaluate if a no-jerk gearshift is possible given a vehicle description, driving scenario, and transmission architecture. This can be used to motivate the choice of a transmission type over another during the conceptual design phase of a new vehicle. For example, if a vehicle is expected to regularly perform upshifts close to saturation in the power-limited region of its motor map, and Theorem~\ref{Thm:OWC} indicates that a gearshift jerk would often be inevitable, then the vehicle designers should consider using a dual-friction-clutch architecture instead of a one-way clutch.  Also, the theorems can be integrated in a transmission control unit: when a gearshift is desired, the unit quickly evaluates if a no-jerk gearshift is possible, and then decides if the driving torque should be smoothly reduced prior to initiating the gearshift, or if the gearshift can be initiated with the current DTD without substantial risks of saturating the motor and obtaining a large driveline jerk. This work also has important implications for academic researchers. The theorems present conditions where gearshift jerk is unavoidable, and any attempt at eliminating jerk with a new controller design would be futile. Moreover, this work helps to identify when two transmission architectures are mathematically equivalent, and therefore will result in the same fundamental limitations.

\section*{Funding sources}
This work was supported by Mitacs and Quebec's Fonds de recherche Nature et technologies.

\bibliographystyle{IEEEtranM}
\bibliography{McGill_PhD}

\end{document}